\documentclass[a4paper,UKenglish,cleveref, autoref]{lipics-v2019}

\usepackage[utf8]{inputenc}
\usepackage[T1]{fontenc}
\usepackage{graphicx}
\usepackage{xspace}

\usepackage{rotating}
\usepackage{tabularx}

\usepackage{mdframed}
\usepackage{tabularx, environ}
\usepackage{cite}
\usepackage{tcolorbox,xspace}
\usepackage{xcolor}
\usepackage[normalem]{ulem}

\usepackage{tikz}

\usetikzlibrary{mindmap,positioning}
\usetikzlibrary{graphs}
\usetikzlibrary[graphs]
\usetikzlibrary{patterns}
\usetikzlibrary{arrows,automata,positioning}
\usetikzlibrary{decorations.pathreplacing}
\usetikzlibrary{decorations}
\usetikzlibrary{shapes.geometric}
\usetikzlibrary{quotes}
\usetikzlibrary{shapes.callouts}
\usetikzlibrary{arrows.meta}
\usetikzlibrary{decorations.pathmorphing}
\usetikzlibrary{calc}

\tikzset{faded/.style={gray,very thin}}
\tikzset{vertex/.style={draw,circle,minimum size=10pt,inner sep=0pt}}
\tikzset{novertex/.style={circle,minimum size=10pt,inner sep=0pt}}
\tikzset{blackvertex/.style={draw,circle,minimum size=10pt,inner sep=0pt, fill=black}}
\tikzset{redvertex/.style={draw,circle,minimum size=10pt,inner sep=0pt, fill=red}}
\tikzset{redvertexfaded/.style={draw,circle,faded,minimum size=10pt,inner sep=0pt, fill=red!50}}
\tikzset{greenvertex/.style={draw,circle,minimum size=10pt,inner sep=0pt, fill=green}}
\tikzset{greenvertexfaded/.style={draw,circle,faded,minimum size=10pt,inner sep=0pt, fill=green!50}}
\tikzset{bluevertex/.style={draw,circle,minimum size=10pt,inner sep=0pt, fill=blue}}
\tikzset{bluevertexfaded/.style={draw,circle,faded,minimum size=10pt,inner sep=0pt, fill=blue!50}}
\tikzset{yellowvertex/.style={draw,circle,minimum size=10pt,inner sep=0pt, fill=yellow}}
\tikzset{yellowvertexfaded/.style={draw,circle,faded,minimum size=10pt,inner sep=0pt, fill=yellow!50}}

\tikzset{edge/.style = {->,> = latex'}}

\tikzset{wavy/.style={decorate, decoration=snake}}

\makeatletter

\newcolumntype{\expand}{}
\long\@namedef{NC@rewrite@\string\expand}{\expandafter\NC@find}

\NewEnviron{boxproblem}[2][]{%
  \def\boxproblem@arg{#1}%
  \def\boxproblem@framed{framed}%
  \def\boxproblem@lined{lined}%
  \def\boxproblem@doublelined{doublelined}%
  \ifx\boxproblem@arg\@empty%
    \def\boxproblem@hline{}%
  \else%
    \ifx\boxproblem@arg\boxproblem@doublelined%
      \def\boxproblem@hline{\hline\hline}%
    \else%
      \def\boxproblem@hline{\hline}%
    \fi%
  \fi%
  \ifx\boxproblem@arg\boxproblem@framed%
    \def\boxproblem@tablelayout{|>{\bfseries}lX|c}%
    \def\boxproblem@title{\multicolumn{2}{|l|}{%
        \raisebox{-\fboxsep}{\textsc{\normalsize #2}}%
      }}%
  \else
    \def\boxproblem@tablelayout{>{\bfseries}lXc}%
    \def\boxproblem@title{\multicolumn{2}{l}{%
        \raisebox{-\fboxsep}{\textsc{\normalsize #2}}%
      }}%
  \fi%
  \bigskip\par\noindent%
  \renewcommand{\arraystretch}{1.2}%
  \begin{tabularx}{\textwidth}{\expand\boxproblem@tablelayout}%
    \boxproblem@hline%
    \boxproblem@title\\[2\fboxsep]%
    \BODY\\\boxproblem@hline%
  \end{tabularx}%
  \medskip\par%
}
\makeatother

\DeclareMathOperator{\width}{\text{width}}
\newcommand{\dtw}{\textsf{dtw}}

\newcommand{\tbd}{\textsc{Disjoint Enough Directed Paths}\xspace}
\newcommand{\tbdp}{\textsc{DEDP}\xspace}
\newcommand{\BcolName}{blocking collection\xspace}

\theoremstyle{plain}
\newtheorem{mycondition}[theorem]{Condition}

\newcommand{\lip}[1]{\textcolor{darkgray}{\sffamily\bfseries #1}}

\hypersetup{
 pdftitle = {A relaxation of the Directed Disjoint Paths problem: a global congestion metric helps},
 colorlinks = true,
 urlcolor = black!50!blue,
 linkcolor = black!50!blue,
 citecolor = black!50!red,
 menucolor = {red}
 filecolor = {cyan},
 anchorcolor = {black}
 urlcolor = {magenta},
 runcolor = {cyan},
 linkbordercolor = {blue}
}

\newcommand{\Ocal}{\mathcal{O}}

\title{A relaxation of the Directed Disjoint Paths problem: a global congestion metric helps}

\titlerunning{A relaxation of the Directed Disjoint Paths problem: a global congestion metric helps}

\author{Raul Lopes}{Departamento de Computa\c c\~ ao, Universidade Federal do Ceará, Fortaleza, Brazil, and \and  LIRMM, Université de Montpellier,  Montpellier, France}{raul@alu.ufc.br}{https://orcid.org/0000-0002-7487-3475
}{Coordenação de Aperfeiçoamento de Pessoal de Nível Superior (CAPES).}

\author{Ignasi Sau}{LIRMM, Universit\'e de Montpellier, CNRS, Montpellier, France}{ignasi.sau@lirmm.fr}{https://orcid.org/0000-0002-8981-9287}{Projects DEMOGRAPH (ANR-16-CE40-0028), ESIGMA (ANR-17-CE23-0010), ELIT (ANR-20-CE48-0008-01) and UTMA (ANR-20-CE92-0027).}

\authorrunning{Raul Lopes and Ignasi Sau}

\Copyright{Raul Lopes and Ignasi Sau}



\category{}

\relatedversion{A conference version of this article appeared in the \emph{Proceedings of the 45th International Symposium on Mathematical Foundations of Computer Science (MFCS), volume 170 of LIPIcs, pages 66:1--66:15, \textbf{2020}}. A full version of the paper is permanently available at \url{https://arxiv.org/abs/1909.13848}.}

\supplement{}

\nolinenumbers 

\hideLIPIcs  

\EventEditors{}
\EventNoEds{2}
\EventLongTitle{}
\EventShortTitle{}
\EventAcronym{}
\EventYear{2019}
\EventDate{December 24--27, 2016}
\EventLocation{Little Whinging, United Kingdom}
\EventLogo{}
\SeriesVolume{42}
\ArticleNo{23}

\begin{document}

\maketitle

\begin{abstract}
In the \textsc{Directed Disjoint Paths} problem, we are given a digraph $D$ and a set of requests $\{(s_1, t_1), \ldots, (s_k, t_k)\}$, and the task is to find a collection of pairwise vertex-disjoint paths $\{P_1, \ldots, P_k\}$ such that each $P_i$ is a path from $s_i$ to $t_i$ in $D$. This problem is \textsf{NP}-complete for fixed $k=2$ and \textsf{W}[1]-hard with parameter $k$ in DAGs. A few positive results are known under restrictions on the input digraph, such as being planar or having bounded directed tree-width, or under relaxations of the problem, such as allowing for vertex congestion. 
Positive results are scarce, however, for general digraphs. In this article we propose a novel global congestion metric for the problem: we only require the paths to be ``disjoint enough'', in the sense that they must behave properly not in the whole graph, but in an unspecified part of size prescribed by a parameter. Namely, in the \tbd  problem, given an $n$-vertex digraph $D$, a set of $k$ requests, and non-negative integers $d$ and $s$,  the task is to find a collection of paths connecting the requests such that at least $d$ vertices of $D$ occur in at most $s$ paths of the collection.
We study the parameterized complexity of this problem for a number of choices of the parameter, including the directed tree-width of $D$. Among other results, we show that the problem is \textsf{W}[1]-hard in DAGs with parameter $d$ and, on the positive side,  we give an algorithm in time $\Ocal(n^{d+2} \cdot k^{d\cdot s})$ and a kernel of size $d \cdot 2^{k-s}\cdot \binom{k}{s} + 2k$ in general digraphs.
This latter result has consequences for the \textsc{Steiner Network} problem: we show that it is \textsf{FPT} parameterized by the number $k$ of terminals and $p$, where $p = n - q$ and $q$ is the size of the solution.

\end{abstract}


\section{Introduction}

In the \textsc{Disjoint Paths} problem, we are given a graph $G$ and a set of pairs of vertices $\{(s_1, t_1), \ldots, (s_k, t_k)\}$, the \emph{requests}, and the task is to find a collection of pairwise vertex-disjoint paths $\{P_1, \ldots, P_k\}$ such that each $P_i$ is a path from $s_i$ to $t_i$ in $G$.
Since this problem is \textsf{NP}-complete in the directed and undirected cases, even if the input graph is planar~\cite{FORTUNE1980111,Lynch:1975:ETP:1061425.1061430}, algorithmic approaches usually involve approximations, parameterizations, and relaxations.
In this article, we focus on the latter two approaches and the directed case.

\medskip
\noindent \textbf{Previous work.}  For the undirected case, Robertson and Seymour~\cite{ROBERTSON199565} showed, in their seminal work on graph minors, that  \textsc{Disjoint Paths} can be solved in time $f(k) \cdot n^{\Ocal(1)}$ for some computable function $f$, where $n$ is the number of vertices of $G$; that is, the problem is \emph{fixed-parameter tractable} (\textsf{FPT}) when parameterized by the number of requests.

The directed case, henceforth referred to as the \textsc{Directed Disjoint Paths}  (\textsc{DDP}) problem, turns out to be significantly harder:  Fortune et al.~\cite{FORTUNE1980111} showed that the problem is \textsf{NP}-complete even for fixed $k=2$. In order to obtain positive results, a common approach has been to consider restricted input digraphs. For instance, it is also shown in~\cite{FORTUNE1980111} that \textsc{DDP} is solvable in time $n^{\Ocal(k)}$ if the input digraph is acyclic.
In other words, \textsc{DDP} is \textsf{XP} in DAGs with parameter $k$.
For some time the question of whether this could be improved to an \textsf{FPT} algorithm remained open, but a negative answer was given by Slivkins~\cite{Slivkins.03}: \textsc{DDP} is \textsf{W}[1]-hard in DAGs with parameter $k$.  Johnson et al. introduced in~\cite{Johnson.Robertson.Seymour.Thomas.01} the notion of directed tree-width, as a measure of the distance of a digraph to being a DAG, and provided generic conditions that, if satisfied by a given problem, yield an \textsf{XP} algorithm on graphs of bounded directed tree-width. In particular, they gave an $n^{\Ocal(k+w)}$ algorithm for \textsc{DDP} on digraphs with directed tree-width at most $w$.  Another restriction considered in the literature is to ask for the underlying graph of the input digraph to be planar. Under this restriction, Schrijver~\cite{Sch94} provided an \textsf{XP} algorithm for \textsc{DDP} with parameter $k$, which was improved a long time afterwards to an \textsf{FPT} algorithm by Cygan et al.~\cite{6686155}.

A natural relaxation  for the \textsc{Directed Disjoint Paths} problem is to allow for vertex and/or edge congestion. Namely, in the \textsc{Directed Disjoint Paths with Congestion} problem (\textsc{DDPC} for short, or \textsc{DDPC-$c$} if we want to specify the value of the congestion), the task is to find a collection of paths satisfying the $k$ requests such that no vertex in the graph occurs in more than $c$ paths of the collection.
Amiri et al.~\cite{Amiri2016RoutingWC} considered the tractability of this problem when restricted to DAGs. 
By a simple local reduction to the disjoint version, they showed how to apply the algorithm by Fortune et al.~\cite{FORTUNE1980111} to solve \textsc{DDPC} in time $n^{\Ocal(k)}$, and proved that, for every fixed $c \geq 1$, not only \textsc{DDPC} is \textsf{W}[1]-hard with relation to the parameter $k$, but also that the exponent $\Ocal(k)$ of $n$ is the best possible under the \emph{Exponential Time Hypothesis}.
Together with the result by Johnson et al.~\cite{Johnson.Robertson.Seymour.Thomas.01}, this simple reduction presented in~\cite{Amiri2016RoutingWC} is sufficient to show that \textsc{DDPC-$c$} admits an \textsf{XP} algorithm with parameters $k$ and $w$ for every fixed $1 \leq c\leq k-1$ in digraphs with directed tree-width at most $w$, and the same result also holds when we allow for congestion on the edges.
In the main algorithmic result of the article, Amiri et al.~\cite{Amiri2016RoutingWC} proved that \textsc{DDPC-$c$} admits an \textsf{XP} algorithm with parameter $d$ in DAGs, where $d = k - c$.

Motivated by Thomassen's proof~\cite{Thomassen91} that \textsc{DDP} remains \textsf{NP}-complete for $k=2$ when restricted to $\beta$-strongly connected digraphs, for any integer $\beta \geq 1$, Edwards et al.~\cite{DBLP:conf/esa/EdwardsMW17} recently considered the \textsc{DDPC-$2$} problem (this version of the problem is usually called \emph{half-integral} in the literature) and proved, among other results, that it can be solved in time $n^{f(k)}$ when restricted to $(36k^3 + 2k)$-strongly connected digraphs.

Kawarabayashi et al.~\cite{KawarabayashiKK14} considered the following asymmetric version of the \textsc{DDPC-$4$} problem: the task is to either find a set of paths satisfying the requests with congestion at most four, or to conclude that no set of pairwise vertex-disjoint paths satisfying the requests exists. In other words, we ask for a solution for \textsc{DDPC-$4$} or a certificate that there is no solution for \textsc{DDP}. They proved that this problem admits an \textsf{XP} algorithm with parameter $k$ in general digraphs, and claimed --without a proof-- that Slivkins' reduction~\cite{Slivkins.03} can be modified to show that it is \textsf{W}[1]-hard in DAGs. In their celebrated proof of the Directed Grid Theorem, Kawarabayashi  and  Kreutzer~\cite{Kawarabayashi:2015:DGT:2746539.2746586} claimed that an \textsf{XP} algorithm can be obtained for the asymmetric version with congestion at most three. To the best of our knowledge, the existence of an \textsf{XP} algorithm in general digraphs for the \textsc{DDPC-$2$} problem, or even for its asymmetric version, remains open.

Summarizing, the existing positive results in the literature for parameterizations and/or relaxations of the \textsc{Directed Disjoint Paths} problem in {\sl general digraphs} are quite scarce.

\medskip
\noindent \textbf{Our approach, results, and techniques.} In this article, we propose another congestion metric for \textsc{DDP}.
In contrast to the usual relaxations discussed above, which focus on a \textsl{local} congestion metric that applies to every vertex, our approach considers, on top of local congestion, a \textsl{global} congestion metric: we want to keep control of how many vertices (a global metric) appear in ``too many'' paths (a local metric) of the solution.
That is, we want the paths to be such that ``most'' vertices of the graph do not occur in too many paths, while allowing for any congestion in the remaining vertices.
In the particular case where we do not allow for local congestion, we want the paths to be pairwise vertex-disjoint not in the whole graph, but in an unspecified part of size prescribed by a parameter; this is why we call such paths ``disjoint enough''.

Formally, in the \tbd~(\tbdp) problem, we are given a set of requests $\{(s_1, t_1), \ldots, (s_k, t_k)\}$ in a digraph $D$ and two non-negative integers $c$ and $s$, and the task is to find a collection of paths $\{P_1, \ldots, P_k\}$ such that each $P_i$ is a path from $s_i$ to $t_i$ in $D$ and at most $c$ vertices of $D$ occur in more than $s$ paths of the collection.
If $s=1$, for instance, we ask for the paths to be pairwise vertex-disjoint in at least $n-c$ vertices of the graph, and allow for at most $c$ vertices occurring in two or more paths.
Choosing $c = 0$ and $s=1$, \tbdp is exactly the \textsc{DDP} problem and, choosing $s=0$, \tbdp is exactly the \textsc{Steiner Network} problem (see~\cite{feldmann_et_al:LIPIcs:2016:6306} for its definition).

Applying simple reductions from \textsc{DDP} and \textsc{DDPC}, we show that \textsc{DEDP} is \textsf{NP}-complete for fixed $k \geq 3$ and $s = 1$, and \textsf{W}[1]-hard in DAGs with parameter $k$, respectively, even if $c$ is large with respect to $n$ in both cases. Namely, if $c$ is at most $n - n^{\alpha}$ for some real value $0 < \alpha \leq 1$.
By applying the framework of Johnson et al.~\cite{Johnson.Robertson.Seymour.Thomas.01}, we give an $n^{\Ocal(k+w)}$ algorithm to solve \textsc{DEDP} in digraphs with directed tree-width at most $w$.

The fact that \tbdp is \textsf{NP}-complete for fixed values of $k=2$, $c=0$, and $s=1$~\cite{FORTUNE1980111} motivates us to consider the ``dual'' parameter $d=n-c$. That is, instead of bounding from above the number of vertices of $D$ that lie in the intersection of many paths of a collection satisfying the given requests, we want to bound from below the number of vertices that occur only in few paths of the collection.
Formally, we want to find  $X \subseteq V(D)$ with $|X| \geq d$ such that there is a collection of paths $\mathcal{P}$ satisfying the given requests such that every vertex in $X$ is in at most $s$ paths of the collection. We first prove, from a reduction from the \textsc{Independent Set} problem, that \tbdp is \textsf{W}[1]-hard with parameter $d$ for every fixed $s \geq 0$,  even if the input graph is a DAG and all source vertices of the request set are the same.

Our main contribution consists of positive algorithmic results for this dual parameterization. On the one hand, we give an algorithm for \tbdp running in time $\Ocal(n^d \cdot k^{d \cdot s})$. This algorithm is not complicated, and basically performs a brute-force search over all vertex sets of size $d$, followed by $k$ connectivity tests in a digraph $D'$ obtained from $D$ by an appropriate local modification. On the other hand, our
 most technically involved result is a kernel for \tbdp with at most
$d\cdot 2^{k-s} \cdot \binom{k}{s}$ non-terminal vertices. This algorithm first starts by a reduction rule that eliminates what we call \emph{congested} vertices; we say that the resulting instance is \emph{clean}. We then show that if $D$ is clean and sufficiently large, and $k = s+1$, then the instance is positive and a solution can be found in polynomial time. This fact is used as the base case of an iterative algorithm.
Namely, we start with the original instance and proceed through $k-s+1$ iterations.
At each iteration, we choose one path from some $s_i$ to its destination $t_i$ such that a large part of the graph remains unused by any of the pairs chosen so far (we prove that such a request always exists) and consider only the remaining requests for the next iteration.
We repeat this procedure until we arrive at an instance where the number of requests is exactly $s+1$, and use the base case to output a solution for it.
From this solution, we extract in polynomial time a solution for the original instance, yielding a kernel of the claimed size.

Since positive results for the \textsc{Directed Disjoint Paths} problem are not common in the literature, especially in general digraphs, we consider our algorithmic results to be of particular interest.
Furthermore, the kernelization algorithm also brings good news for the \textsc{Steiner Network} problem: when $s = 0$
Feldmann and Marx in~\cite{feldmann_et_al:LIPIcs:2016:6306} showed that the tractability of the \textsc{Steiner Network} problem when parameterized by the number of requests depends on how the requests are structured.
Our result adds to the latter by showing that the problem remains \textsc{FPT} if we drop this structural condition on the request set but add $d$, the number of vertices occurring in at most $s$ paths of the solution, as a parameter.
More details can be found in Section~\ref{sec:defs}.

Table~\ref{table:summary_of_results} shows a summary of our algorithmic and complexity results, which altogether provide an accurate picture of the parameterized complexity of the \tbdp problem for distinct choices of the parameters.

\begin{table}[h!tb]
\centering
\begin{tabular}{|c|c|c|c|l|}
\hline
$k$   & $d$ & $s$ & $w$ & Complexity\\ \hline
 fixed $\geq 3$ & $ \Omega(n^\alpha)$ & fixed $ = 1 $ & \---- & \textsf{NP}-complete (Theorem~\ref{theorem:npc_for_k})\\ \hline
  parameter & $ \Omega(n^\alpha)$ & fixed $ \geq 1 $ & 0 & \textsf{W}[1]-hard (Theorem~\ref{theorem:npc_for_k})\\ \hline
  input & parameter & fixed $\geq 0$ & \---- & \textsf{W}[1]-hard (Theorem~\ref{theorem:w1_hardness_for_d})\\ \hline
parameter  & \---- & \---- & parameter & \textsf{XP} (Theorem~\ref{thm:XPdtw}) \\ \hline
input  & parameter & parameter & \---- & \textsf{XP} (Theorem~\ref{algorithm:xp_d_s_1})\\ \hline
parameter  & parameter & parameter & \---- & \textsf{FPT} (Theorem~\ref{corollary:kernel_for_k_d_s_problem})\\ \hline
\end{tabular}
\medskip
\caption{Summary of hardness and algorithmic results for distinct choices of the parameters. A horizontal line in a cell means no restrictions for that case, and here we denote by $w$ the directed tree-width of the input digraph. In all cases, we have that $c=n-d$.}
\label{table:summary_of_results}
\end{table}%

\vspace{-.5cm}
\noindent \textbf{Organization.} In Section~\ref{sec:defs} we present some preliminaries relevant to all parts of this article and formally define the  \textsc{Disjoint Enough Directed Paths} problem.
We provide the hardness results in Section~\ref{section:lower_bounds} and the algorithms in Section~\ref{sec:algo}. 
The corresponding notations and definitions related to directed tree-width are presented in Section~\ref{subsection:algorithm_dtw}, where they are used.
We conclude the article in Section~\ref{subsection:poly_kernel_open} with some open questions for further research.

\section{Preliminaries and definitions}
\label{sec:defs}

For a graph $G = (V,E)$, directed or not, and a set $X \subseteq V(G)$, we write $G - X$ for the graph resulting from the deletion of $X$ from $G$ and $G[X]$ for the graph induced by $X$.
We also write $G' \subseteq G$ to say that $G'$ is a subgraph of $G$.
If $e$ is an edge of a directed or undirected graph with endpoints $u$ and $v$, we may refer to $e$ as $(u,v)$.
If $e$ is an edge of a digraph, we say that $e$ has \emph{tail} $u$, $\emph{head}$ $v$ and is \emph{oriented} from $u$ to $v$.

The \emph{in-degree} $\deg^-_D(v)$ (resp. \emph{out-degree} $\deg^+_D(v))$ of a vertex $v$ in a digraph $D$ is the number of edges with head (resp. tail) $v$.
The \emph{degree} $\deg_D(v)$ of $v$ in $D$ is the sum of $\deg_D^-(v)$ with $\deg_D^+(v)$.
The \emph{in-neighborhood} $N^-_D(v)$ of $v$ is the set $\{u \in V(D) \mid (u,v) \in E(G)\}$, and the \emph{out-neighborhood} $N^+_D(v)$ is the set $\{u \in V(D) \mid (v,u) \in E(G)\}$.
We say that $u$ is an \emph{in-neighbor} of $v$ if $u \in N^-_D(v)$ and that $u$ is an \emph{out-neighbor} of $v$ if $u \in N^+_D(v)$.

A \emph{walk} in a digraph $D$ is an alternating sequence $W$ of vertices and edges that starts and ends with a vertex, and such that for every edge $(u,v)$ in the walk, vertex $u$ (resp. vertex $v$) is the element right before (resp. right after) edge $(u,v)$ in $W$. A walk is a \emph{path} if all the vertices in it are  distinct. All paths mentioned henceforth, unless stated otherwise, are considered to be directed.

An \emph{orientation} of an undirected graph $G$ is a digraph $D$ obtained from $G$ by choosing an orientation for each edge $e \in E(G)$.
The undirected graph $G$ formed by ignoring the orientation of the edges of a digraph $D$ is the \emph{underlying graph} of $D$.

A digraph $D$ is \emph{strongly connected} if, for every pair of vertices $u,v \in V(D)$, there is a walk from $u$ to $v$ and a walk from $v$ to $u$ in $D$.
We say that $D$ is \emph{weakly connected} if the underlying graph of $D$ is connected.
A \emph{separator} of $D$ is a set $S \subsetneq V(D)$ such that $D \setminus S$ is not strongly connected.
If $|V(D)| \geq k+1$ and $k$ is the minimum size of a separator of $D$, we say that $D$ is \emph{$k$-strongly connected}.
A \emph{strong component} of $D$ is a maximal induced subdigraph of $D$ that is strongly connected, and a \emph{weak component} of $D$ is a maximal induced subdigraph of $D$ that is weakly connected.

Unless stated otherwise, $n$ will always denote the number of vertices of the input graph.
For an integer $\ell \geq 1$, we denote by $[\ell]$ the set $\{1, 2, \ldots, \ell\}$.
We make use of Menger's Theorem~\cite{Menger1927} for digraphs.
Here a $(u,v)$-\emph{separator} is a set of vertices $X$ such that there is no path from $u$ to $v$ in $D-X$.
\begin{theorem}[{Menger's Theorem}]
Let $D$ be a digraph and $u,v \in V(D)$ such that $(u,v) \not \in E(D)$.
Then the minimum size of a $(u,v)$-separator equals the maximum number of pairwise internally vertex-disjoint paths from $u$ to $v$.
\end{theorem}

\subsection{Parameterized complexity}

We refer the reader to~\cite{DF13,CyganFKLMPPS15} for basic background on parameterized complexity, and we recall here only some basic definitions. A \emph{parameterized problem} is a language $L \subseteq \Sigma^* \times \mathbb{N}$.  For an instance $I=(x,k) \in \Sigma^* \times \mathbb{N}$, $k$ is called the \emph{parameter}.

A parameterized problem is \emph{fixed-parameter tractable} (\textsf{FPT}) if there exists an algorithm $\mathcal{A}$, a computable function $f$, and a constant $c$ such that given an instance $I=(x,k)$, $\mathcal{A}$   (called an \textsf{FPT} \emph{algorithm}) correctly decides whether $I \in L$ in time bounded by $f(k) \cdot |I|^c$. For instance, the \textsc{Vertex Cover} problem parameterized by the size of the solution is \textsf{FPT}.
	
A parameterized problem is \textsf{XP} if there exists an algorithm $\mathcal{A}$ and two computable functions $f$ and $g$ such that given an instance $I=(x,k)$, $\mathcal{A}$  (called an \textsf{XP} \emph{algorithm}) correctly decides whether $I \in L$ in time bounded by $f(k) \cdot |I|^{g(k)}$. For instance,  the \textsc{Clique} problem parameterized by the size of the solution is in  \textsf{XP}.

Within parameterized problems, the \textsf{W}-hierarchy may be seen as the parameterized equivalent to the class \textsf{NP} of classical decision problems. Without entering into details (see~\cite{DF13,CyganFKLMPPS15} for the formal definitions), a parameterized problem being \textsf{W}[1]-\emph{hard} can be seen as a strong evidence that this problem is \textsl{not} \textsf{FPT}.
The canonical example of \textsf{W}[1]-hard problem is \textsc{Clique}  parameterized by the size of the solution.

For an instance $(x, k)$ of a parameterized problem $Q$, a \emph{kernelization algorithm} is an algorithm $\mathcal{A}$ that, in polynomial time, generates from $(x, k)$ an equivalent instance $(x', k')$ of $Q$ such that $|x'| + k' \leq f(k)$, for some computable function $f : \mathbb{N} \rightarrow \mathbb{N}$. If $f(k)$ is bounded from above by a polynomial of the parameter, we say that $Q$ admits a \emph{polynomial kernel}.

A \emph{polynomial time and parameter reduction} is an algorithm  that, given an instance $(x, k)$ of a parameterized problem $A$, runs in time $f(k)\cdot |x|^{\Ocal(1)}$ and outputs an instance $(x', k')$ of a parameterized problem $B$ such that $k'$ is bounded from above by a polynomial on $k$ and $(x, k)$ is positive if and only if $(x',k')$ is positive. 

\subsection{The \textsc{Disjoint Enough Directed Paths} problem}\label{section:routing_and_related_problems}

Before defining the problem, we define requests and satisfying collections.
\begin{definition}[Requests and satisfying collections]\label{definition:requests}
Let $D$ be a digraph and $\mathcal{P}$ be a collection of paths  of $D$.
A \emph{request} in $D$ is an ordered pair of vertices of $D$.
For a request (multi)set $R = \{(s_1,t_1), (s_2, t_2), \ldots, (s_k, t_k)\}$, we say that the vertices $\{s_1, s_2, \ldots, s_k\}$ are \emph{source} vertices and that $\{t_1, t_2, \ldots, t_k\}$ are \emph{target} vertices, and we refer to them as $S(R)$ and $T(R)$, respectively.
We say that $\mathcal{P}$ \emph{satisfies} $I$ if $\mathcal{P} = \{P_1, \ldots, P_k\}$ and $P_i$ is a path from $s_i$ to $t_i$, for $i \in [k]$.
\end{definition}
We remark that a request multiset may contain many copies of the same pair, and that when considering the union of two or more of those multisets, we keep all such copies in the resulting request multiset.
For instance, if $R_1 = \{(u_1,v_1)\}$ and $R_2 = \{(u_1, v_1), (u_2, v_2)\}$ then $R_1 \cup R_2 = \{(u_1, v_1), (u_1, v_1), (u_2, v_2)\}$, and this indicates that a collection of paths satisfying this request set must contain two paths from $u_1$ to $v_1$.
To simplify the notation, we simply refer to request multisets as request sets.
The \tbdp~problem is defined as follows.\vspace{-.2cm}

\begin{boxproblem}[framed]{\tbd (\tbdp)}\label{problem:main_problem_c}
Input: & A digraph $D$, a request set $R$ of size $k$, and two non-negative integers $c$ and $s$.\\
Output: & A collection of paths $\mathcal{P}$ satisfying $R$ such that at most $c$ vertices of $D$ occur in at least $s+1$ paths of $\mathcal{P}$ and all other vertices of $D$ occur in at most $s$ paths of $\mathcal{P}$.
\end{boxproblem}
Unless stated otherwise, we consider $d= n -c$ for the remaining of this article.
Intuitively, $c$ imposes an upper bound on the size of the ``\textbf{c}ongested'' part of the solution, while $d$ imposes a lower bound on the size of the ``\textbf{d}isjoint'' part.
For a parameterized version of \tbdp, we sometimes include the parameters before the name.
For instance, we denote by $(k,d)$-\tbdp the \tbd~problem with parameters $k$ and $d$.
We refer to instances of \textsc{DEDP} as $(D, R, k, c, s)$.

Notice that if $c \geq n$ or $s \geq k$, the problem is trivial since every vertex of the graph is allowed to be in all paths of a collection satisfying the requests, and thus we only need to check for connectivity between the given pairs of vertices.
Furthermore, if there is a pair $(s_i,t_i)$ in the request set such that there is no path from $s_i$ to $t_i$ in the input digraph $D$, the instance is negative.
Thus we henceforth assume that $c < n$, that $s < k$, and that there is a path from $s_i$ to $t_i$ in $D$ for every pair $(s_i,t_i)$ in the set of requests.

Choosing the values of $k, d$, and $s$ appropriately, we show in Table~\ref{table:summary_related_problems} that the \tbdp problem generalizes some problems in the literature.

\begin{table}[ht] 
\renewcommand{\arraystretch}{1.5}
\centering
\begin{tabular}{|m{2cm}|m{4cm}|m{5cm}|}
\hline
Parameters & Equivalent to  & Complexity\\ \hline
$d = n$, $s = 1$ & \textsc{Directed Disjoint Paths}& \textsf{NP}-complete for $k =2$~\cite{FORTUNE1980111}\\ \hline
$d = n$, $s \geq 1$ & \textsc{Directed Disjoint Paths} \textsc{with Congestion} $s$ & \textsf{W}[1]-hard with parameter $k$~\cite{Slivkins.03,Amiri2016RoutingWC}\\ \hline
$d \geq 1$, $s = 0$ & \textsc{Steiner Network} & \textsf{FPT} with parameters $k$ and $d$\\ \hline
\end{tabular}
\medskip
\caption{Summary of related problems.}\vspace{-.2cm}
\label{table:summary_related_problems}
\end{table}%

The last line of Table~\ref{table:summary_related_problems} is of particular interest, and we focus on it in the next two paragraphs.
In the \textsc{Steiner Network} problem, we are given a digraph $D$ and a request set $R$ and we are asked to find an induced subgraph $D'$ of $D$ with minimum number of vertices such that $D'$ admits a collection of paths satisfying $R$.
For a request set $R$ in a digraph $D$, let $D(R)$ be the digraph with vertex set $S(R) \cup T(R)$ and edge set $\{(s,t) \mid (s,t) \in R\}$.
The complexity landscape of the \textsc{Steiner Network} problem when parameterized by the size of the request set was given by Feldmann and Marx~\cite{feldmann_et_al:LIPIcs:2016:6306}.
They showed that the tractability of the problem depends on $D(R)$.
Namely, they proved that if $D(R)$ is close to being a \emph{caterpillar}, then the \textsc{Steiner network} problem is \textsf{FPT} when parameterized by $|R|$, and \textsf{W}[1]-hard otherwise.
When parameterized by the size of the solution, Jones et al.~\cite{10.1007/978-3-642-40450-4_57} showed that the \textsc{Steiner Network} problem is \textsf{FPT} when $D(R)$ is a star whose edges are all oriented from the unique source and the underlying graph of the input digraph excludes a topological minor, and \textsf{W}$[2]$-hard on graphs of degeneracy two~\cite{10.1007/978-3-642-40450-4_57}.

Our algorithmic results for \tbdp for the particular case $s=0$ yield an \textsf{FPT} algorithm for another parameterized variant of the \textsc{Steiner Network} problem.
In this case, we want to decide whether $D$ admits a large set of vertices whose removal does not disconnect any pair of requests.
That is, we want to find a set $X \subseteq V(D)$ with $|X| \geq d$ such that $D - X$ contains a collection of paths satisfying $R$.
In Theorem~\ref{corollary:kernel_for_k_d_s_problem} we give an \textsf{FPT} algorithm  (in fact, a kernel) for this problem with parameters $|R|$ and $d$.
We remark that this tractability does not depend on $D(R)$.

%
\section{Hardness results for \textsc{DEDP}}\label{section:lower_bounds}
In this section we provide hardness results for the \tbdp problem. Namely, we first provide in Theorem~\ref{theorem:npc_for_k} a simple reduction from \textsc{Disjoint Paths with Congestion}, implying \textsf{NP}-completeness for fixed values of $k,c,d$ when $s=1$ and \textsf{W}[1]-hardness in DAGs with parameter $k$ when $s \geq 1$ and $s < k$. We then prove in Theorem~\ref{theorem:w1_hardness_for_d}
that \tbdp is \textsf{W}[1]-hard in DAGs with parameter $d$.

\medskip

As mentioned in~\cite{10.1007/978-3-642-40450-4_57}, the \textsc{Steiner Network} problem is \textsf{W}$[2]$-hard when parameterized by the size of the solution (as a consequence of the results of~\cite{Molle2008}).
Hence $(c)$-\tbdp is \textsf{W}$[2]$-hard for fixed $s = 0$. As discussed in the introduction, the \textsc{Directed Disjoint Paths} problem is \textsf{NP}-complete for fixed $k=2$~\cite{FORTUNE1980111} and \textsf{W}[1]-hard with parameter $k$ in DAGs~\cite{Slivkins.03}.
In addition, \textsc{Directed Disjoint Paths with Congestion} parameterized by the number of requests is also \textsf{W}[1]-hard in DAGs for every fixed congestion $s \geq 1$, as observed in~\cite{Amiri2016RoutingWC}.
When $c=0$ and $s\geq 1$, \tbdp is equivalent to the \textsc{Directed Disjoint Paths with Congestion} problem and thus the aforementioned bounds apply to it as well.
In the following theorem we complete this picture  by showing that \tbdp is \textsf{NP}-complete for fixed $k \geq 3$ and $s =1$, even if $c$ is quite large with respect to $n$ (note that if $c=n$ all instances are trivially positive), namely for
$c$ as large as $n - n^\alpha$ with $\alpha$ being any fixed real number such that $0 < \alpha \leq 1$. The same reduction also allows to prove \textsf{W}[1]-hardness in DAGs with parameter $k$.
The idea is, given the instance of \textsc{DDPC} with input digraph $D$, build an instance of \tbdp where the ``disjoint'' part corresponds to the original instance, and the ``congested'' part consists of $c$ new vertices that are necessarily used by $s+1$ paths of any solution.
In this process, we generate an instance of \textsc{DEDP} in a digraph $D'$ with $|V(D')| = n = d + c$ and $d = |V(D)|$.
This is why we restrict the value of $d$ to be of the form $n^\alpha$, but not smaller: if we ask $d$ to be ``too small'', for example $d = \log n$, our procedure would generate an instance of \textsc{DEDP} such that the size of the ``disjoint part'' $d$ satisfies $d = \log(d + c)$ which in turn implies that the size of this instance would be exponential on the size of the original instance of \textsc{DDPC}.

Following, we refer to instances of \textsc{DDPC} with input graph $D$, request set $R$, $k = |R|$, and congestion $s$ as $(D, R, k, s)$.
We remind the reader that we can assume that $s < k$, since otherwise the problem reduces to a simple connectivity check between every pair of vertices in the request set.
 
\begin{theorem}\label{theorem:npc_for_k}
Let $0 < \alpha \leq 1$, $d: \mathbb{N} \to \mathbb{N}$ with ${\bf d}(n) = \Omega(n^\alpha)$, and ${\bf c}(n) = n - {\bf d}(n)$.
Then, for $c = {\bf c}(n)$ and $d = {\bf d}(n)$,
\begin{romanenumerate}
	\item  {\sc DEDP} is $\mathsf{NP}$-complete for every fixed $k \geq 3$ and $s = 1$; and
	\item  $(k)$-{\sc DEDP} is $\mathsf{W[1]}$-hard in DAGs for every fixed $s \geq 1$.
\end{romanenumerate}
\end{theorem}
\begin{proof}
We prove items \textbf{(i)} and \textbf{(ii)} at the same time by a simple reduction from the \textsc{Directed Disjoint Paths with Congestion} (\textsc{DDPC})~problem.
Given an instance $(D,R, k,s)$ of \textsc{DDPC}, we output an equivalent instance $(D', R', k+s, c, s)$ of \tbdp that does not generate any new cycles and such that the size $d(|V(D')|)$ of the disjoint part of the new instance is equal to $|V(D)|$, with ${\bf d}(n)$ as in the statement of the theorem.
Since \textsc{DDP}, which is exactly the \textsc{DDPC} problem with congestion $s = 1$,
is \textsf{NP}-complete for fixed $k \geq 2$~\cite{FORTUNE1980111} and $k$-\textsc{DDPC} is \textsf{W}[1]-hard in DAGs~\cite{Amiri2016RoutingWC}, our reduction implies that \tbdp with $c = n - {\bf d}(n)$ is \textsf{NP}-complete for every fixed $k\geq 3$ and $s = 1$, and \textsf{W}[1]-hard in DAGs with parameter $k$ and any fixed $s \geq 1$.
We can assume that $c \geq 1$ since \tbdp is exactly \textsc{DDPC} when $c = 0$ and $s \geq 1$ (as discussed previously).

Formally, let $(D, R, k, s)$ be an instance of \textsc{DDPC} with $R = \{(s_1, t_1), (s_2, t_2), \ldots, (s_k, t_k)\}$ and choose $i \in [k]$ arbitrarily.
We construct an instance of \tbdp as follows.
Let $D'$ be a digraph constructed by adding to $D$ a path with vertex set $\{v_1, \ldots, v_c\}$ and an edge from $t_i$ to $v_1$.
Then, add to $R'$ all pairs in $(R \setminus \{s_i, t_i\})$, the pair $(s_i, v_c)$, and $s$ copies of the pair $(v_1, v_{c})$.
Figure~\ref{fig_NP_completeness_proof} illustrates this construction.
\begin{figure}[h!]
\centering
\scalebox{1}{
\begin{tikzpicture}[square/.style={regular polygon,regular polygon sides=4}]
\draw[rounded corners] (0,-0.15) rectangle  (5,2.15) node [above,xshift=-2.5cm] {$V(D)$};
\node[square,draw, fill=black,scale=.5,label=below:$s_2$] (P-s2) at (0.5,0.5) {};
\node[square,draw, fill=black,label=above:$s_1$,scale=.5] (P-s1) at (0.5,1.5) {};
\node[blackvertex,scale=.5,label=below:$t_2$] (P-t2) at (4.5,0.5) {};
\node[square,draw, fill=black,scale=.5,label=above:$t_1$] (P-t1) at (4.5,1.5) {};
\node[square,draw, fill=black,scale=.5,label=right:$v_c$] (P-vc) at (9,0.5) {};

\node[square,draw, fill=black,scale=.5, label=above:$v_1$] (P-v1) at (P-t2) [xshift=2cm] {};
\draw[-{Latex[length=2mm, width=2mm]}, shorten >= 1] (P-t2) -- (P-v1);
\draw[-{Latex[length=2mm, width=2mm]}, shorten >= 1] (P-v1) -- (P-vc);
\foreach \i in {1,...,6}{
	\node[blackvertex,scale=.25] (P-x\i) at (P-v1) [xshift=2*\i cm] {};
}
\draw [decorate,decoration={mirror,brace,amplitude=10pt}]
($(P-v1.south) - (0,.2cm)$) -- ($(P-vc.south) - (0, .2cm)$) node [black,midway,yshift=-0.6cm] {$\{v_1, \ldots, v_c\}$};

\draw [-{Latex[length=2mm, width=2mm]},decorate,decoration={snake,amplitude=4mm,segment length=10mm,post length=4mm}, shorten >= 1] (P-s1) -- (P-t1);
\draw [-{Latex[length=2mm, width=2mm]},decorate,decoration={snake,amplitude=4mm,segment length=10mm,post length=2mm}, shorten >= 1] (P-s2) -- (P-t2);

\end{tikzpicture}%
}%
\caption[Example of the construction from Theorem~\ref{theorem:npc_for_k}]{Example of the construction from Theorem~\ref{theorem:npc_for_k} with $k=2$, $s=1$, and $i=2$. Source and target vertices are represented by square vertices in the figure.}
\label{fig_NP_completeness_proof}
\end{figure}
It is easy to verify that $(D, R, k, s)$ is positive if and only if the instance $(D', R', k+s, c, s)$ of \tbdp is positive since every solution to the second contains $s+1$ copies of the path from $v_1$ to $v_c$ (one being to satisfy the pair $(s_i, v_c)$), every vertex in $V(D)$ can occur in at most $s$ paths of any solution.

Since $D'$ is formed by adding a path on $c$ vertices to a copy of $D$, it is constructed in time $\Ocal(|V(D)| + |E(D)| + c)$.
Choosing $c$ to satisfy $d(|V(D)'|) = |V(D)|$, the hypothesis that $d(|V(D)'|)=\Omega(|V(D)'|^\alpha)$ easily implies that $c = \Ocal(|V(D)|^{1/\alpha})$ which in turn implies that the procedure ends in polynomial time.\end{proof}

In Theorem~\ref{theorem:npc_for_k} we only prove \textsf{NP}-completeness for $s = 1$ since, as we discuss in Section~\ref{subsection:poly_kernel_open}, it is not known if \textsc{DDPC} is \textsc{NP}-complete for some fixed $k$ when the congestion is at least two.
Nevertheless, a positive answer to this question would also imply the \textsf{NP}-completeness of \textsc{DEDP} for some fixed $k$ and fixed $s \geq 1$ as the reduction used in the proof of Theorem~\ref{theorem:npc_for_k} implies that \textsc{DEDP} is as hard as \textsc{DDPC}.
Dealing with the local congestion metric $s$, however, is not the main objective of this article, which is to focus on the global congestion metric $c$ and its dual $d = n - c$.

Next, we show that $(d)$-\tbdp is \textsf{W}[1]-hard, even when the input graph is acyclic and all source vertices of the request set are the same.
The reduction is from the \textsc{Independent Set} problem  parameterized by the size of the solution, which is \textsf{W}[1]-hard~\cite{CyganFKLMPPS15,DF13}.

\begin{theorem}\label{theorem:w1_hardness_for_d} The {\sc DEDP} problem is $\mathsf{W}[1]$-hard with parameter $d$  for every fixed $s \geq 0$, even when the input graph is acyclic and all source vertices in the request set are the same.
\end{theorem}
\begin{proof}
Let $(G, d)$ be an instance of the \textsc{Independent Set} problem, in which we want to decide whether the (undirected) graph $G$ contains an independent set of size at least $d$, and $s$ be a non-negative integer.
Let $V^E$ be the set $\{v_{e} \mid e \in E(G)\}$ and $D$ a directed graph with vertex set $V(G) \cup \{r\} \cup V^E$.
Add to $D$ the following edges:
\begin{itemize}
	\item[$\bullet$] for every $v \in V(G)$, add the edge $(r,v)$; and
	\item[$\bullet$] for every edge $e \in E(G)$ with endpoints $u$ and $w$, add the edges $(u,v_{e})$ and $(w,v_{e})$.
\end{itemize}
Finally, for every $v_e \in V^E$, add $2s+1$ copies of the pair $(r,v_e)$ to $R$.
Figure~\ref{fig:construction_w1_hardness}  illustrates this construction.
\begin{figure}[htb]
\centering
\scalebox{0.5}{
\begin{tikzpicture}[scale=2]
	\node[blackvertex,scale=1,label=left:{\huge $u$}] (P-u) at (-2,-2) {};

	\node[blackvertex,scale=1,label=right:{\huge $w$}] (P-w) at (2,-2) {};

	\node[blackvertex,scale=1,label=above:{\huge $r$}] (P-r) at (0,-0.5) {};

	\node[blackvertex,scale=1,label=below:{\huge $v_e$}] (P-e) at (0,-2) {};

\draw[-{Latex[length=3mm, width=2mm]}] (P-r) to [out = 180, in = 90]  (P-u);
\draw[-{Latex[length=3mm, width=2mm]}] (P-r) to [out = 0, in = 90]  (P-w);

\foreach \i in {25}{
	\draw[-{Latex[length=3mm, width=2mm]}, dashed] (P-r) to [bend left=2*\i]   (P-e);
	\draw[-{Latex[length=3mm, width=2mm]}, dashed] (P-r) to [bend right=2*\i]   (P-e);
}

\draw[-{Latex[length=3mm, width=2mm]}, dashed] (P-r) to (P-e);
\foreach \x/\y in {
u/e, w/e}
{
	\draw[-{Latex[length=3mm, width=2mm]}] (P-\x) to   (P-\y);
}
\end{tikzpicture}%
}%
\caption{Example of the construction from Theorem~\ref{theorem:w1_hardness_for_d} with $s=1$ and $e = (u,w)$. A dashed line indicates a request in $R$.}
\label{fig:construction_w1_hardness}
\end{figure}

Notice that each vertex $v_e$ of $D$ associated with an edge $E$ of $D$ has out-degree zero in $D$ and  $r$ has in-degree zero.
Moreover, every edge of $D$ has as extremity either $r$ or a vertex of the form $v_e$.
Thus $D$ is a acyclic, as desired.
Furthermore $|S(R)| = 1$ since all of its elements are of the form $(r, v_e)$, for $e \in E(G)$.
We now show that $(G,d)$ is positive if and only if $(D, R, k, c, s)$ is positive, where $k = |R| = m\cdot(2s+1)$.

For the necessity, let $X$ be an independent set of size $d$ in $G$.
Start with a collection $\mathcal{P} = \emptyset$.
We classify the edges of $G$ into two sets: the set $E_1$ containing all edges with both endpoints in $V(G) - X$, and the set $E_2$ containing all edges with exactly one endpoint in $X$.
Now, for each $e \in E_1$, choose arbitrarily one endpoint $u$ of $e$ and add to $\mathcal{P}$ $2s+1$ copies of the path in $D$ from $r$ to $v_e$ using $u$.
For each $e \in E_2$ with $e = (u,w)$ and $w \not \in X$, add to $\mathcal{P}$ $2s+1$ copies of the path in $D$ from $r$ to $v_e$ using $w$.
Since $X$ is an independent set, no vertex in $X$ occurs in any path of $\mathcal{P}$, and since $E(G) = E_1 \cup E_2$, $\mathcal{P}$ satisfies $R$ and the necessity follows as $c = n - d$.

Let $\mathcal{P}$ be a solution for $(D, R, k, c, s)$ and $X \subseteq V(D)$ be a set of vertices with $|X| = d$ and such that each vertex of $X$ occurs in at most $s$ paths of $\mathcal{P}$. Such choice is possible since $d = n - c$.
For contradiction, assume that $X$ is not an independent set in $G$.
Then there is an edge $e \in E(G)$ with $e=(u,w)$ and $u,w \in X$, and $2s+1$ copies of the request $(r, v_e)$ in $R$.
Thus each path satisfying one of those requests uses $u$ or $w$, but not both, and therefore either $u$ or $w$ occurs in at least $s+1$ paths of $\mathcal{P}$, a contradiction.
We conclude that $X$ is an independent set in $G$ and the sufficiency follows.
\end{proof}
%
\section{Algorithms for DEDP}
\label{sec:algo}
In this section we focus on algorithmic results for \tbdp.
In Theorem~\ref{theorem:npc_for_k} we showed that  \tbdp is \textsf{NP}-complete for every fixed $k \geq 3$ and a large range of values of $c$, and showed that considering only $d$ as a parameter is still not enough to improve the tractability of the problem: Theorem~\ref{theorem:w1_hardness_for_d} states that $(d)$-\tbdp is \textsf{W}[1]-hard in DAGs even if all requests share the same source.
Thus \tbdp is as hard as the \textsc{Directed Disjoint Paths} problem when $k$ is \textit{not} a parameter.
Here we show that, similarly to the latter, \tbdp admits an \textsf{XP} algorithm when parameterized by the number of requests and the directed tree-width of the input digraph.
Then, we show how the tractability of \tbdp improves when we consider stronger parameterizations including $d$.
Namely, we  show that \tbdp is \textsf{XP} with parameters $d$ and $s$ (cf. Theorem~\ref{algorithm:xp_d_s_1}), and \textsf{FPT} with parameters $k$ and $d$ (hence $s$ as well, since we may assume that $k > s$ as discussed in Section~\ref{section:routing_and_related_problems}; cf. Theorem~\ref{corollary:kernel_for_k_d_s_problem}).
It is worth mentioning that this kind of \emph{dual parameterization} (remember that $d = n - c$) has proved useful in order to improve the tractability of several notoriously hard problems (cf. for instance~\cite{ChorFJ04,BasavarajuFRS16,DuhF97,AraujoCLSSS18,Amiri2016RoutingWC}).

In Section~\ref{subsection:algorithm_dtw} we formally define directed tree-width and arboreal decompositions, as provided by Johnson et al.~\cite{Johnson.Robertson.Seymour.Thomas.01}, and apply the ideas and results they used to show an \textsf{XP} algorithm for the \textsc{Directed Disjoint Paths} problem parameterized by the number of requests in digraphs of bounded directed tree-width to show that a similar result holds for \tbdp.
In Section~\ref{section:algorithm:xp_d_s_1} we show our algorithms for parameterizations of \tbdp including $d$ as a parameter.

\subsection{An \texorpdfstring{$\mathsf{XP}$}{XP} algorithm with parameters \texorpdfstring{$k$}{k} and \texorpdfstring{$\mathsf{dtw}(D)$}{dtw(D)}}\label{subsection:algorithm_dtw}
Given the success obtained in the design of efficient algorithms in undirected graphs of bounded tree-width (cf.~\cite{Cook.Seymour.2003,COURCELLE199012}, for example), and  the enormous success achieved by the Grid Theorem~\cite{Robertson.Seymour.86} and by the \emph{Bidimensionality} framework~\cite{Demaine:2005:SPA:1101821.1101823}, it is no surprise that there was interest in finding an analogous definition for digraphs.
As the tree-width of an undirected graph measures, informally, its distance to being a tree, the \emph{directed tree-width} of a digraph, as defined by Johnson et al.~\cite{Johnson.Robertson.Seymour.Thomas.01}, measures its distance to being a DAG, and an \emph{arboreal decomposition} of a digraph exposes a (strong) connectivity measure of the original graph.
The authors conjectured the existence of a grid-like theorem for their width measure~\cite{Johnson.Robertson.Seymour.Thomas.01}.
In a recent breakthrough, Kawarabayashi and Kreutzer~\cite{Kawarabayashi:2015:DGT:2746539.2746586} proved this conjecture to be true: they showed that there is a computable function $g$ such that every digraph of directed tree-width at least $g(k)$ has a \emph{cylindrical grid} as a \emph{butterfly minor}\footnote{The full version of~\cite{Kawarabayashi:2015:DGT:2746539.2746586},  available at https://arxiv.org/abs/1411.5681, contains all these definitions.}. Recently Campos et al.~\cite{LAGOS19} improved the running time of the algorithm that follows from the proof of Kawarabayashi and Kreutzer~\cite{Kawarabayashi:2015:DGT:2746539.2746586}, by locally modifying some steps of the original proof.

The technical contents of this section are mostly taken from~\cite{Johnson.Robertson.Seymour.Thomas.01}.
By an \emph{arborescence} $T$, we mean an orientation of a tree with root $r_0$ in such a way that all edges are pointing away from $r_0$.
If a vertex $v$ of $T$ has out-degree zero, we say that $v$ is a \emph{leaf} of $T$.
We now define guarded sets and arboreal decompositions of directed graphs.
From here on, we refer to oriented edges only, unless stated otherwise, and
$D$ will always stand for a directed graph.
All the considered directed graphs mentioned may contain directed cycles of length two.
\begin{definition}[$Z$-guarded sets]\label{def:guarded_sets}
Let $D$ be a digraph, let $Z \subseteq V(D)$, and $S \subseteq V(D) \setminus Z$.
We say that $S$ is \emph{$Z$-guarded} if there is no directed walk in $D - Z$ with first and last vertices in $S$ that uses a vertex of $D - (Z \cup S)$.
For an integer $w \geq 0$, we say that $S$ is \emph{$w$-guarded} if $S$ is $Z$-guarded for some set $Z$ with $|Z| \leq w$.
\end{definition}
See Figure~\ref{fig:z-guarded} for an illustration of a $Z$-guarded set.
\begin{figure}[h!tb]
\centering
\scalebox{.75}{
\begin{tikzpicture}[scale=1, -{Latex[length=2mm, width=2mm]}, shorten >= .15cm, shorten <= .1cm]
\draw[rounded corners] (0,4) rectangle  (5,5) node [above,xshift=-2.5cm] {$V(D)\setminus (Z \cup S)$} ;
\draw[rounded corners] (0,0) rectangle (5,1) node [below,yshift=-1cm,xshift=-2.5cm] {$S$}; 
\draw[rounded corners] (1,2) rectangle (4,3) node [above,xshift=.5cm,yshift=-.75cm] {$Z$};
\node[blackvertex,scale=.5] (P-a) at (0.5,0.5) {};
\node (P-b) at (-0.5,2.5) {};
\node (P-c) at (0.25,4.5) {};
\node (P-d) at (0.85,4.75) {};
\node[blackvertex,label=210:{\Large$u$},scale=.5] (P-e) at (1.25,4.5) {};
\node (P-f) at (1.5,2.5) {};
\node (P-g) at (0.75,1.5) {};
\node[blackvertex,scale=.5] (P-h) at (1.5,0.5) {};

\draw[thick] plot[smooth] coordinates {(P-a) (P-b) (P-c) (P-d) (P-e) (P-f) (P-g) (P-h) };

\node[blackvertex,scale=.5] (P-a1) at (4,0.5) {};
\node (P-b1) at (5,2.5) {};
\node (P-c1) at (4.25,4.5) {};
\node (P-d1) at (3.75,4.75) {};
\node[blackvertex,label=-30:{\Large$v$},scale=.5] (P-e1) at (2.5,4.5) {};
\node (P-f1) at (3,2.5) {};
\node (P-g1) at (2.25,1.5) {};
\node[blackvertex,scale=.5] (P-h1) at (3,0.5) {};

\draw[thick] plot[smooth] coordinates {(P-h1) (P-g1) (P-f1) (P-e1) (P-d1) (P-c1) (P-b1) (P-a1)};

\draw[dashed] (P-e) to [bend right =30] (P-e1);

\end{tikzpicture}%
}%
\caption{A $Z$-guarded set $S$. The dashed line indicates that there cannot be a walk from $u$ to $v$ in $V(D)\setminus (Z \cup S)$.}
\label{fig:z-guarded}
\end{figure}
If a set $S$ is $Z$-guarded, we may also say that $Z$ is a \emph{guard} for $S$. We remark that in~\cite{Johnson.Robertson.Seymour.Thomas.01,Kawarabayashi:2015:DGT:2746539.2746586}, the authors use the terminology $Z$-\emph{normal} sets instead of $Z$-guarded sets.
In this article, we adopt the terminology used, for instance, in~\cite{Classes.Directed.Graphs}.

Let $T$ be an arborescence, $r \in V(T)$, $e \in E(T)$, and $r'$ be the head of $e$.
We say that $r > e$ if there is a path from $r'$ to $r$ in $T$ (notice that this definition implies that $r' > e$).
We say that $e \sim r$ if $r$ is the head or the tail of $e$.
To define the directed tree-width of directed graphs, we first need to introduce arboreal decompositions.
\begin{definition}[Arboreal decomposition and directed tree-width]
An \emph{arboreal decomposition} $\beta$ of a digraph $D$ is a triple $(T,\mathcal{X},\mathcal{W})$ where $T$ is an arborescence, $\mathcal{X} = \{X_e : e \in E(T)\}$, $\mathcal{W} = \{W_r : r \in V(T)\}$, and $\mathcal{X},\mathcal{W}$ are collections of sets of vertices of $D$ (called \emph{bags}) such that
\begin{enumerate}
\item[\textbf{(i)}] $\mathcal{W}$ is a partition of $V(D)$ into non-empty sets, and
\item[\textbf{(ii)}] if $e \in E(T)$, then $\bigcup\{W_r : r \in V(T) \text{ and } r > e\}$ is $X_e$-guarded.
\end{enumerate}
We say that $r$ is a \emph{leaf} of $(T,\mathcal{X,W})$ if $r$ has out-degree zero in $T$.

For a vertex $r \in V(T)$, we denote by $\width(r)$ the size of the set $W_r \cup (\bigcup_{e \thicksim r}X_e)$.
The \emph{width} of $(T,\mathcal{X},\mathcal{W})$ is the least integer $k$ such that, for all $r \in V(T)$, $\width(r) \leq k+1$.
The \emph{directed tree-width}  of $D$, denoted by $\dtw(D)$, is the least integer $k$ such that $D$ has an arboreal decomposition of width $k$.
\end{definition}
The left hand side of Figure~\ref{fig:arboreal_decomposition} contains an example of a digraph $D$, while the right hand side shows an arboreal decomposition for it.
In the illustration of the arboreal decomposition, squares are guards $X_e$ and circles are bags of vertices $W_r$.
For example, consider the edge $e \in E(T)$ with $X_e = \{b,c\}$ from the bag $W_1$ to the bag $W_2$.
Then $\bigcup\{W_r : r \in V(T) \text{ and } r > e\} = V(D) \setminus \{a\}$ and, by item \textbf{(ii)} described above, this set must be $\{b,c\}$-guarded since $X_e = \{b,c\}$.
In other words, there cannot be a walk in $D \setminus \{b,c\}$ starting in $V(D) \setminus \{a\}$ going to $\{a\}$ and then back. 
This is true in $D$ since every path reaching $\{a\}$ from the remaining of the graph must do so through vertices $b$ or $c$.
The reader is encouraged to verify the same properties for the other guards in the decomposition.
\begin{figure}[ht]
\centering
\begin{subfigure}{.35\textwidth}
\scalebox{.85}{
\begin{tikzpicture}

\foreach \x/\y/\name/\lpos in {
	0/5/a/90,
	-1.5/4/b/90, 1.5/4/c/90}
	\node[blackvertex, label=\lpos:{$\name$},scale=.5] (P-\name) at (\x,1.2*\y) {$\name$};

\foreach \x/\y/\name/\lpos in {
	-2.25/3/d/120, -0.75/3/e/60, 0.75/3/f/120, 2.25/3/g/60}
	\node[blackvertex, label={[label distance= -.1cm]\lpos:{$\name$}},scale=.5] (P-\name) at (\x,1.2*\y) {$\name$};

\foreach \x/\y in {
a/b, a/c, b/c, b/d, b/e, c/b, c/f, c/g, d/b, e/b, f/c, g/c, d/e, f/g}
	\draw[edge, {Latex[length=2mm, width=2mm]}-{Latex[length=2mm, width=2mm]}, line width = 1, shorten >= .1cm, shorten <= .1cm] (P-\x) to (P-\y);

\end{tikzpicture}%

}%
\end{subfigure}\hspace{1cm}
\begin{subfigure}{.35\textwidth}
\scalebox{.75}{
\begin{tikzpicture}
	\foreach \x/\y/\name/\idn/\lbl/\lblpos in {
	0/5/w1/a/{W_1}/180,
	0/2.5/w2/{b,c}/{W_2}/270,
	-2.5/2/w3/{d,e}/{W_3}/90, 2.5/2/w4/{f,g}/{W_4}/90}
	\node[vertex, text width=.75cm, align=center, label, label=\lblpos:{\large$\lbl$}] (P-\name) at (\x,\y) {\idn};

\draw[edge, line width = 1] (P-w1) to  node[midway, rectangle, fill=white,draw] {$b,c$}  (P-w2); 
\draw[edge, line width = 1] (P-w2) to  node[midway, rectangle, fill=white,draw] {$b$}  (P-w3); 
\draw[edge, line width = 1] (P-w2) to  node[midway, rectangle, fill=white,draw] {$c$}  (P-w4); 

\end{tikzpicture}%
}%
\end{subfigure}
\caption{A digraph $D$ and an arboreal decomposition of $D$ of width two. A bidirectional edge is used to represent a pair of edges in both directions.}
\label{fig:arboreal_decomposition}
\end{figure}

We remark that DAGs have directed tree-width zero.
As shown in~\cite{Johnson.Robertson.Seymour.Thomas.01}, it is not hard to see that, if $D$ is a digraph constructed by replacing every edge of an undirected graph $G$ by two edges in opposite directions, then the directed tree-width of $D$ is equal to the tree-width of $G$.
This observation also implies that $\dtw(D)$ is at most the tree-width of its underlying graph.
Further intuitions of the similarities between the undirected and the directed cases are given by Reed~\cite{Reed.99}.
It is worth noting that, in contrast to the undirected case, the class of digraphs of bounded directed tree-width is \emph{not} closed under butterfly contractions~\cite{ADLER2007718}.

The algorithm we present in this section consists of dynamic programming along an arboreal decomposition of the input digraph. Following the notation used by Johnson et al.~\cite{Johnson.Robertson.Seymour.Thomas.01}, we refer to the information we want to compute at every step of the algorithm as an \emph{itinerary}.
We provide a formal definition for an itinerary for \tbdp later.
We recall that a set of vertices $S$ is $w$-guarded if $S$ is $Z$-guarded for some $Z$ with $|Z| \leq w$ (cf. Definition~\ref{def:guarded_sets}).

Johnson et al.~\cite{Johnson.Robertson.Seymour.Thomas.01} provided two conditions that, if satisfied by a given problem, are sufficient to provide an \textsf{XP} algorithm for it in digraphs with bounded directed tree-width.
More precisely, for a digraph $D$ with $\dtw(D) = w$, they ask that there is a real number $\alpha$ depending on $w$ (and possibly some parameters of the problem, if any) and two algorithms satisfying the following conditions.

\begin{mycondition}[Johnson et al.~{\cite{Johnson.Robertson.Seymour.Thomas.01}}]\label{condition_1}
Let $A,B$ be two disjoint subsets of $V(D)$ such that there are no edges in $D$ with head in $A$ and tail in $B$.
Then an itinerary for $A \cup B$ can be computed from an itinerary for $A$ and an itinerary for $B$ in time $\Ocal(n^{\alpha})$.
\end{mycondition}
\begin{mycondition}[Johnson et al.~{\cite{Johnson.Robertson.Seymour.Thomas.01}}]\label{condition_2}
Let $A,B$ be two disjoint subsets of $V(D)$ such that $A$ is $w$-guarded and $|B| \leq w$.
Then an itinerary for $A \cup B$ can be computed from an itinerary for $A$ and an itinerary for $B$ in time $\Ocal(n^{\alpha})$.
\end{mycondition}
Using this notation, the following theorem says how to compute an itinerary for $V(D)$.
\begin{theorem}[Johnson et al.~\cite{Johnson.Robertson.Seymour.Thomas.01}]\label{theorem:meta_theorem_dtw}
Provided that Conditions~\ref{condition_1} and~\ref{condition_2} hold, there is an algorithm running in time $\Ocal(n^{\alpha +1})$ that receives as input a digraph $D$ and an arboreal decomposition for $D$ with width at most $w$ and outputs an itinerary for $V(D)$.
\end{theorem}

In~\cite{Johnson.Robertson.Seymour.Thomas.01} an \textsf{XP} algorithm for the \textsc{Directed Disjoint Paths} problem in digraphs of bounded directed tree-width is given as an example of application of the aforementioned tools, and a similar approach is claimed to work for the \textsc{Hamilton Path, Hamilton Path With Prescribed Ends, Even Cycle Through a Specified Vertex} problems, and others.
We follow their ideas to provide an \textsf{XP} algorithm for $(k,w)$-\tbdp, where $w$ is the directed tree-width of the input digraph.
The main idea, formalized by the following definition and lemma, is that the number of weak components in the digraph formed by the union of the paths in a collection $\mathcal{P}$ satisfying the request set is bounded by a function depending on $k$ and $w$ only.
Thus we can guess how the paths in $\mathcal{P}$ cross a set of vertices $A$ that is $w$-guarded and use an arboreal decomposition of the input digraph to propagate this information in a dynamic programming scheme.
We use the following definition.
\begin{definition}
Let $D$ be a digraph and $\mathcal{P}$ be a collection of paths in $D$.
We denote by $D(\mathcal{P})$ the digraph formed by the union of all paths in $\mathcal{P}$.
\end{definition}

\begin{definition}[Limited collections]
Let $R$ be a request set in a digraph $D$ with $|R| = k$ and $\mathcal{P}$ be a collection of paths satisfying $R$.
We say that $\mathcal{P}$ is \emph{$(k,w,S)$-limited}, for some $S \subseteq V(D)$, if $D(\mathcal{P})\subseteq D[S]$ and for every $w$-guarded set $S' \subseteq S$, the digraph induced by $V(D(\mathcal{P})) \cap S'$ has at most $(w+1) \cdot k$ weak components.
\end{definition}
The following lemma is inspired by~\cite[Lemma 4.5]{Johnson.Robertson.Seymour.Thomas.01} and is key to the algorithm.
\begin{lemma}\label{lem:limited_collections}
Let $R$ be a request set of size $k$ in a digraph $D$ and $w$ be an integer.
Then every collection of paths $\mathcal{P}$ satisfying $R$ is $(k, w, S)$-limited for every $S \subseteq V(D)$ containing all paths in $\mathcal{P}$.
\end{lemma}
\begin{proof}
Let $k = |R|$ and $S$ be as in the statement of the lemma and $S'$ be a $w$-guarded subset of $S$.
By the definition of $w$-guarded sets, there is a set $Z \subseteq V(D)$ with $|Z| \leq w$ such that $S'$ is $Z$-guarded.
For $i \in [k]$, let  $\mathcal{Q}_i$ be the collection of paths formed by the subpaths of $P_i$ intersecting~$S'$.
Thus, $D(\mathcal{Q}_i)$ consists of the union of subpaths of $P_i$.
Let $q_i$ be the number of weak components of $D(\mathcal{Q}_i)$.
Since $S'$ is $Z$-guarded, each subpath of $P_i$ linking two distinct weak components of $D(\mathcal{Q}_i)$ must intersect $Z$.
Thus, $|V(P_i) \cap Z| \geq q_i - 1$ and $q_i \leq w+1$ since a vertex of $Z$ can be in all paths of $\mathcal{P}$.
We conclude that
$\sum_{i \in [k]}q_i \leq (w+1) \cdot k$,
as desired.
\end{proof}

We now formally define an itinerary for \tbdp.
From this point forward, we say that a request set $R$ in a digraph $D$ is \emph{contained} in $A$ if every vertex occurring in $R$ is contained in $A$.
\begin{definition}[Itinerary]\label{def:itinerary}
Let $\Gamma$ be an instance of {\sc DEDP} with $\Gamma = (D, R, k, c, s)$, $A \subseteq V(D)$, and $\mathcal{R}_A$ be the set of all request sets on $D$ which are contained in $A$.
For an integer $w$, a \emph{$(\Gamma,w)$-itinerary for $A$} is a function $f_A : \mathcal{R}_A \times \mathbb{N} \rightarrow \{0,1\}$ such that $f_A(R', c') = 1$ if and only if
\begin{enumerate}
	\item[\emph{(i)}] $k' \leq (w+1) \cdot k$, for $k' = |R'|$;
	\item[\emph{(ii)}] $c' \leq c$; and
	\item[\emph{(iii)}] the instance $(D[A], R', k', c', s)$ of {\sc DEDP} is positive.
\end{enumerate}
\end{definition}

With this notation, an instance $(D, R, k, c, s)$ is positive if and only if $f_{V(D)}(R, c') = 1$ for some $c' \leq c$.
We now provide algorithms satisfying Conditions~\ref{condition_1} and~\ref{condition_2} for the given definition of an itinerary for \tbdp.
By Lemma~\ref{lem:limited_collections}, we need to consider only request sets of size at most $(w+1)\cdot k$ whenever the input digraph has directed tree-width at most $w$ in the following lemmas.
We follow the proofs given by Johnson et al.~\cite{Johnson.Robertson.Seymour.Thomas.01}, adapting them to our case.
For every $t \in [n]$, the authors show how to compute a solution containing at most $t$ vertices for a given instance of the \textsc{Directed Disjoint Paths} problem, if one exists, or to decide that no such solution exists.
We drop this demand in our algorithm, and instead include the restriction on the congestion $c$.

\begin{lemma}\label{lemma:condition_1_proof}
Let $\Gamma$ be an instance of {\sc DEDP} with $\Gamma = (D, R, k, c, s)$ and $A,B$ be disjoint subsets of $V(D)$ such that there are no edges in $D$ with head in $A$ and tail in $B$.
Then a $(\Gamma,w)$-itinerary for $A \cup B$ can be computed from itineraries for $A$ and $B$ in time $\Ocal(n^{4(w+1) \cdot k + 3})$.
\end{lemma}
\begin{proof}
Let $f_A$ and $f_B$ be $(\Gamma,w)$-itineraries for $A$ and $B$, respectively.
Given a request set $L$ contained in $A \cup B$, with $L = \{(s_1, t_1), \ldots, (s_\ell, t_\ell)\}$ and $\ell \leq (w+1) \cdot k$, and an integer $c' \in [c]$, we show how to correctly define $f_{A \cup B}(L, c')$ by looking at $f_A$, $f_B$, and the edges in $D$ from $A$ to $B$.

If $L$ is contained in $A$ we set $f_{A \cup B}(L, c') = f_A(L, c')$ as there are no edges from $B$ to $A$ in $D$, and set $f_{A \cup B}(L, c') = f_B(L, c')$ if $L$ is contained in $B$.
If there is a pair $(s,t) \in L$ such that $s \in B$ and $t \in A$, we set $f_{A\cup B}(L, c') = 0$. 
Assume now that no such pairs exist in $L$ and that $L$ is not contained in $A$ nor in $B$.

Define $L_A = L_B = \emptyset$.
For $i \in [\ell]$, do the following:
\begin{enumerate}
	\item If $s_i \in A$ and $t_i \in A$, define $s^A_i = s_i$, $t^A_i = t_i$ and include the pair $(s^A_i, t^A_i)$ in $L_A$.
	\item If $s_i \in B$ and $t_i \in B$, define $s^B_i = s_i$, $t^B_i = t_i$ and include the pair $(s^B_i, t^B_i)$ in $L_B$.
	\item If $s_i \in A$ and $t_i \in B$, define $s^A_i = s_i$, $t^B_i = t_i$, choose $t^A_i \in A$ and $s^B_i \in B$ arbitrarily in such way that there is an edge from $t^A_i$ to $s^B_i$ in $D$, include $(s^A_i, t^A_i)$ in $L_A$ and $(s^B_i, t^B_i)$ in $L_B$.
\end{enumerate}
Figure~\ref{fig:construction_condition_1} illustrates this construction.
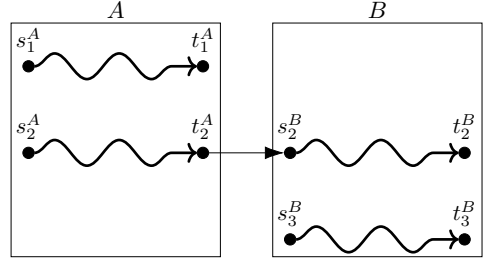
\begin{figure}[htb]
\centering
\scalebox{.85}{
\begin{tikzpicture}[scale=1.35,-{Latex[length=2mm, width=2mm]}, shorten >=.1cm]
\node[blackvertex, label={\Large$s^A_1$},scale=.5] (P-a) at (0,2) {};
\node[blackvertex, label={\Large$t^A_1$},scale=.5] (P-b) at (2,2) {};

\node[blackvertex, label={\Large$s^A_2$},scale=.5] (P-c) at (0,1) {};
\node[blackvertex, label={\Large$t^A_2$},scale=.5] (P-d) at (2,1) {};
\node[blackvertex, label={\Large$s^B_2$},scale=.5] (P-e) at (3.2,1) {};
\node[blackvertex, label={\Large$t^B_2$},scale=.5] (P-f) at (5.2,1) {};

\node[blackvertex, label={\Large$s^B_3$},scale=.5] (P-g) at (3.2,0) {};
\node[blackvertex, label={\Large$t^B_3$},scale=.5] (P-h) at (5.2,0) {};

\foreach \x/\y in {
a/b, c/d, e/f, g/h}
	\draw [decorate, decoration={snake, segment length=10mm, amplitude=2mm},  line width=1.2] (P-\x) -- (P-\y);

\foreach \x/\y in {
d/e}
	\draw (P-\x) -- (P-\y);

\draw (-0.2,-0.2) rectangle (2.4,2.75);
\draw (2.8,-0.2) rectangle (5.4,2.75);
\node (P-n) at (1,3) {\Large$A$};
\node (P-m) at (4,3) {\Large$B$};
\end{tikzpicture}%
}%
\caption{Example of the construction from Lemma~\ref{lemma:condition_1_proof}.}
\label{fig:construction_condition_1}
\end{figure}
\noindent Now, for $c' \in [c]$, if $f_A(L_A, c_1) = f_B(L_B, c_2) = 1$ for some $c_1$, $c_2$ with $c_1 + c_2 \leq c'$, then we set $f_{A \cup B}(L, c') = 1$.
Otherwise, we repeat the procedure used to construct $L_A$ and $L_B$ with a different choice for $t^A_i$ and/or $s^B_i$ in the third step. If all possible choices of $L_A, L_B, c_1$, and $c_2$ have been considered this way, we set $f_{A \cup B}(L,c') = 0$.
We now show that this definition of $f_{A \cup B}(L,c')$ is correct.

Consider the instance $(D[A\cup B], L, \ell, c', s)$ of \tbdp, let $k_1 = |L_A|$, and $k_2 = |L_B|$.
If it is positive, then for some choice of $L_A$, $L_B$, $c_1$, and $c_2$, there are collections of paths $\mathcal{P}_A$ and $\mathcal{P}_B$ and integers $c_1$ and $c_2$ such that $\mathcal{P}_A$  and $\mathcal{P}_B$ are solutions for the instances $(D[A], L_A, k_1, c_1, s)$ and $(D[B], L_B, k_2, c_2, s)$ of \tbdp, respectively. Thus, $f_A(L_A, c_1) = f_B(L_B, c_2) = 1$ since $L_A$ and $L_B$ are at most as large as $L$.
Conversely, if the above equation holds for some choice of $L_A, L_B, c_1$, and $c_2$ such that $c_1 + c_2 \leq c'$, we can construct a solution for the instance on $D[A \cup B]$ by considering the union of a solution for $(D[A], L_A, k_1, c_1, s)$ with a solution for $(D[B], L_B, k_2, c_2, s)$, together with edges from targets of $L_A$ to sources of $L_B$ which where considered in  step~\lip{3} described above.
We conclude that the instance $(D[A\cup B], L, \ell, c', s)$ is positive if and only if there are integers $c_1$, $c_2$ and request sets $L_A$, $L_B$ such that $c_1 +  c_2 \leq c$ and
$$f_A(L_A, c_1) = f_B(L_B, c_2) = 1.$$

By definition, to compute a $(\Gamma,w)$-itinerary for $A \cup B$ we need to consider every request set of size at most $(w+1) \cdot k$ contained  $A \cup B$.
Thus there are $n^{2(w+1) \cdot k}$ choices of $L$.
By construction, there are at most $n^{2(w+1\cdot k)}$ choices for $L_A$ and $L_B$ in total since we need to choose only the vertices $t_i^A$ and $s_i^B$ in step~\lip{3} of the construction of those requests sets.
Finally, since $c' \in [c]$ and $c < n$, the bound on the running time follows.
\end{proof}

\begin{lemma}\label{lemma:condition_2}
Let $\Gamma$ be an instance of {\sc DEDP} with $\Gamma = (D, R, k, c, s)$ and $A,B \subseteq V(D)$ such that $A$ is $w$-guarded and $|B| \leq w$.
Then a $(\Gamma,w)$-itinerary for $A \cup B$ can be computed from $(\Gamma,w)$-itineraries for $A$ and $B$ in time $\Ocal(n^{4(w+1) \cdot k + 2})$.
\end{lemma}
\begin{proof}
Let $f_A$ be a $(\Gamma, w)$-itinerary for $A$, $L$ be a request set contained in $A \cup B$ with $L = \{(s_1, t_1), \ldots, (s_\ell, t_\ell)\}$ and $\ell \leq (w+1)\cdot k$, and $\Gamma_{c'}$ be the instance $(D[A \cup B], L, \ell, c', s)$ of \tbdp, for $c' \in [c]$.

For each pair $(s,t) \in L$, a path from $s$ to $t$ in $D[A \cup B]$ in a solution for $\Gamma_{c'}$ may be entirely contained in $A$, entirely contained in $B$, or it may intersect both $A$ and $B$.
We can test if there is a solution for $\Gamma_{c'}$ whose paths are all contained in $A$ by verifying the value of $f_A(L,c')$, and in time $\Ocal(2^{w \cdot \ell})$ we can test if there is a solution for $\Gamma_{c'}$ that is entirely contained in $B$, since $|B| \leq w$.
We now consider the case where all solutions for $\Gamma_{c'}$ contain a path intersecting both $A$ and $B$.

Suppose that $P$ is a path from $s$ to $t$ in a solution $\mathcal{P}$ for $\Gamma_{c'}$.
Let $\mathcal{P}_A$ be the set of subpaths of $P$ which are contained in $A$, with $\mathcal{P}_A = \{P^{A}_1, \ldots, P^{A}_a\}$, and, for $i \in [a]$, let $u_i$ and $v_i$ be the first and last vertices occurring in $P^{A}_i$, respectively.
Furthermore, let $\mathcal{P}_B$ be the collection of subpaths of $P$ contained in $B \cup (\bigcup_{i \in[a]}\{u_i, v_i\})$.
Then $\mathcal{P}_B$ is a collection of disjoint paths satisfying the request set $\{(v_1, u_2), \ldots, (v_{a-1}, u_a)\}$, together with $(s,u_1)$ if $s \in B$ and $(v_a, t)$ if $t \in B$, such that each path of $\mathcal{P}_B$ has its extremities in $B \cup \{u_1, v_1, \ldots, u_a, v_a\}$ and all internal vertices in $B$.
Figure~\ref{fig:construction_condition_2_example} illustrates this case.
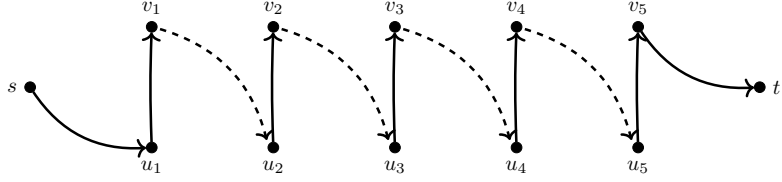
\begin{figure}[t]
\centering
\scalebox{.8}{
\begin{tikzpicture}[scale=2,-{Latex[length=2mm, width=2mm]}, shorten >=.1cm]
\foreach \i in {1,...,5}{
	\node[blackvertex, label=below:{\Large$u_{\i}$},scale=.5] (P-u\i) at (\i,-0.5) {};
	\node[blackvertex, label=above:{\Large$v_{\i}$},scale=.5] (P-v\i) at (\i,0.5) {};
	\draw[line width=1.2] (P-u\i) to [bend left =2] (P-v\i);
}

\edef\j{0}%
\foreach \i in {1,...,4}{
	\pgfmathparse{\i+1}%
	\edef\j{\pgfmathresult}%
	\node (n) at (P-u\j) {};
	\draw[line width=1.2, dashed, {Latex[length=2mm, width=2mm]}-, shorten <=.1cm]  (n) to [in = 0, out=180]  (P-v\i);
}

\node[blackvertex, label=left:{\Large$s$},scale=.5] (P-s) at (0, 0) {};
\node[blackvertex, label=right:{\Large$t$},scale=.5] (P-t) at (6, 0) {};

\draw[line width=1.2] (P-s) to [bend right =30] (P-u1);
\draw[line width=1.2] (P-v5) to [bend right =30] (P-t);

\end{tikzpicture}%
}%
\caption{Collections $\mathcal{P}_A$ and $\mathcal{P}_B$. A continuous line represents a piece of $P$ contained in $A$ and a dashed line represents a piece of $P$ contained in $B$.}
\label{fig:construction_condition_2_example}
\end{figure}

The number of such collections is a function depending on $a$ and $w$ only and, by Lemma~\ref{lem:limited_collections} and our assumption that $A$ is $w$-guarded, we can assume that $a \leq (w+1)\cdot k$.
We show how we can test whether there is a solution for $\Gamma_{c'}$ using an itinerary for $A$ and, for each $(s,t) \in L$, searching for a collection $\mathcal{P}_B$ as described above.
Intuitively, we want to guess how the paths in a solution for $\Gamma_{c'}$ intersect $A$ and how those pieces can be connected through $B$.

For $i \in [\ell]$, let $L_i = \{(u^i_1, v^i_1), \ldots, (u^i_{\ell_i}, v^i_{\ell_i})\}$ and $L_A = L_1 \cup L_2 \cup \cdots \cup L_\ell$ (keeping each copy of duplicated entries) such that
\begin{enumerate}
	\item $u^i_1 = s_i$ and $v^i_{\ell_i} = t_i$, for $i \in [\ell]$;
	\item all vertices occuring in $L_i$ are in $A$ except possibly $u^i_1$ and $v^i_{\ell_i}$ (which may occur in $A \cup B$); and
	\item $|L_A| \leq (w+1) \cdot k$.
\end{enumerate}

By Lemma~\ref{lem:limited_collections}, we only need to consider request sets of size at most $(w+1)\cdot k$ in $A \cup B$ since every solution for $\Gamma$ has at most $(w+1)\cdot k$ weak components in $A$.
Let $B^+$ be the set formed by the union of $B$ with all vertices occurring in $L_A$ and in
$$L_B = \bigcup_{i \in [\ell]} \{(v^i_j,u^i_{j+1}) \mid j \in [\ell_i - 1]\}.$$
That is, for each pair $(u^i_j,v^i_{j}) \in L_A$ that we want to satisfy in $A$, we want to link this subpath in a (possible) solution for $\Gamma_{c'}$ to the next one through a path in $B^+$ satisfying the pair $(v^i_{j},u^i_{j+1}) \in L_B$.
We claim that there is a solution for $\Gamma_{c'}$ if and only if, for some choice of $L_A$, $c_1$, and $L_B$, we have $f_A(L_A, c_1) = 1$ and a collection of paths $\mathcal{P}_B$ satisfying $L_B$ in $D[B^+]$ such that
\begin{enumerate}
	\item[(a)] every path of $\mathcal{P}_B$ starts and ends in $B^+ - B$ and has all of its internal vertices in $B$; and
	\item[(b)] at most $c - c_1$ vertices of $B$ occur in more than $s$ paths of $\mathcal{P}_B$.
\end{enumerate}

For the necessity, we can choose $L_A$ and $L_B$ as described above in this proof.
For the sufficiency, let $\ell_A = \sum_{i \in [\ell]}\ell_i$ and  $\mathcal{P}_A$ be a solution for the instance
$(D[A], L_A, \ell_A, c_1, s)$ of \tbdp, with $\mathcal{P}_A = \{P_1, \ldots, P_{\ell_A}\}$.
Now, since the paths in this collection are not necessarily disjoint, we are guaranteed to find only a directed walk from $s_i$ to $t_i$ for each pair $(s_i, t_i) \in L$ by linking (through the paths in $B^+$) the endpoints of the paths in the collection satisfying $L_i$, with $i \in [\ell]$.
However, every such directed walk contains a path from $s_i$ to $t_i$ whose set of vertices is contained in the set of vertices of the walk.
Thus by following those directed walks and choosing the paths appropriately, we can construct a solution for $\Gamma_{c'}$, since shortening the walks can only decrease the number of vertices occurring in $s+1$ or more paths of the collection.

The number of collections $\mathcal{P}_B$ for which (a) and (b) hold is $\Ocal(2^{w\cdot\ell})$ and thus depending on $k$ and $w$ only, since $\ell \leq (w+1)\cdot k$. Since $|A| \leq n$, $c \leq n$, and the number of itineraries contained in $A \cup B$ is at most $n^{4(w+1)\cdot k}$, the bound on the running time follows.
\end{proof}

Finally, we obtain the \textsf{XP} algorithm combining Lemmas~\ref{lemma:condition_1_proof} and~\ref{lemma:condition_2} together with Theorem~\ref{theorem:meta_theorem_dtw}.
\begin{theorem}\label{thm:XPdtw}
The {\sc DEDP} problem is solvable in time $\Ocal(n^{4(w+1)\cdot k +3})$ in digraphs of directed tree-width at most $w$.
\end{theorem}

\subsection{Algorithms for the dual parameterization}\label{section:algorithm:xp_d_s_1}
We now show our algorithmic results for stronger parameterizations of \tbdp including $d$ as a parameter.
The following definition is used in the description of the algorithms of this section.

\begin{definition}
Let $D$ be a digraph, $R$ be a request set with $R = \{(s_1, t_1), \ldots, (s_k, t_k)\}$, and $s$ be an integer.
We say that a set $X \subseteq V(D)$ is \emph{$s$-viable for $R$} if there is a collection of paths $\mathcal{P}$ satisfying $R$ such that each vertex of $X$ occurs in at most $s$ paths of $\mathcal{P}$.
We also say that $\mathcal{P}$ is \emph{certifying $X$}.
\end{definition}

Thus an instance $(D, R, k, c, s)$ of \tbdp is positive if and only if $D$ contains an $s$-viable set $X$ with $|X| \geq d$.
In other words, we want to find a set of vertices $X$ of size at least $d$ such that there is a collection of paths $\mathcal{P}$ satisfying $R$ that is ``well-behaved'' inside of $X$; that is, the paths of $\mathcal{P}$ may intersect freely outside of $X$, but each vertex of $X$ must be in at most $s$ paths of $\mathcal{P}$.
When $s=1$, for instance, instead of asking for the paths to intersect only inside a small set of vertices (size at most $c$), we ask for them to be disjoint inside a large set of vertices (size at least $d$).
Since we now consider $d$ as a parameter instead of $c$, from this point onwards we may refer to instances of \tbdp as $(D, R, k, d, s)$.

The following definition is needed in the next proof.
For two positive integers $a$ and $b$ with $a \geq b$, the \emph{Stirling number of the second kind}~\cite{math.mag.85.4.252}, denoted by $\textsf{Stirling}(a,b)$, counts the number of ways to partition a set of $a$ objects into $b$ non-empty subsets, and is bounded from above by
$\frac{1}{2} \binom{a}{b}\cdot b^{a-b}$.

\begin{theorem}\label{algorithm:xp_d_s_1}
There is an algorithm running in time $\Ocal(n^{d+2} \cdot k^{d\cdot s})$ for the {\sc Disjoint Enough Directed Paths} problem.
\end{theorem}
\begin{proof}
Let $D$ be a graph on $n$ vertices and $(D, R, k, d,s)$ be an instance of \tbdp.
Notice that if $X$ is $s$-viable for $R$, then any proper subset of $X$ is $s$-viable for $R$ as well.
Therefore, we can restrict our attention to sets of size exactly $d$.

If $s=0$, it is sufficient to test whether there is a $d$-sized set $X \subseteq V(D)$ such that there is a collection of paths satisfying $R$ in $D - X$, and this can be done in time $\Ocal(n^d \cdot k \cdot n^2)$.

Let now $s=1$, and  $R = \{(s_1, t_1), \ldots, (s_k, t_k)\}$.
We claim that a set $X \subseteq V(D)$ is $1$-viable for $R$ if and only if there is a partition $\mathcal{X}$ of $X$ into sets $X_1, \ldots, X_k$ such that $X \setminus X_i$ is not an $(s_i, t_i)$-separator, for $i \in [k]$.

Let $\mathcal{X}$ be as stated in the claim. For each $i \in [k]$, let $P_i$ be a path from $s_i$ to $t_i$ in $D - (X \setminus X_i)$.
Now, $\{P_1, \ldots, P_k\}$ is a collection satisfying $R$ and no pair of paths in it intersect inside $X$.
Thus $X$ is $1$-viable for $R$ as desired.
For the necessity, since $X$ is $1$-viable for $R$, there is a collection of paths $P_1, \ldots, P_k$ satisfying $R$ such that $V(P_i) \cap V(P_j) \cap X = \emptyset$ for all $i,j \in [k]$ with $i \neq j$.
Thus we choose $\mathcal{X} = \{X_1, \ldots, X_k\}$ with $X_i = V(P_i) \cap X$ for $i \in [k-1]$ and $X_k = X \setminus (X_1, \ldots, X_{k-1})$ and the claim follows.

Assume now that $X \subseteq V(D)$ with $|X| = d$.
By the previous claim, we can check whether $X$ is $1$-viable for $R$ by testing whether $X$ admits a partition into (possibly empty) sets $X_1, \ldots, X_k$ such that $X \setminus X_i$ is not an $(s_i, t_i)$-separator.
Since $\textsf{Stirling}(d,k) = \Ocal(k^d)$, this yields an  algorithm in time $\Ocal(n^{d+2} \cdot k^d)$ for the \tbdp problem when $s = 1$.

For $s \geq 2$, let $X = \{v_1, \ldots, v_d\}$ and construct a graph $D'$ from $D$ by making $s$ copies $v_i^1, \ldots, v_i^s$ of each vertex $v_i \in X$ and adding one edge from each copy to each vertex in the neighborhood of $v_i$ in $D$, respecting orientations.

For $i \in [d]$, let $V_i = \{v_i^1, \ldots, v_i^s\}$ and $X' = \bigcup_{i \in [d]}V_i$.
Now, there is a collection of paths $\mathcal{P}$ satisfying $R$ in $D$ such that each vertex in $X$ is in at most $s$ paths of $\mathcal{P}$ if and only if there is a collection of paths $\mathcal{P'}$ in $D'$ such that no vertex in $X'$ occurs in more than one path of $\mathcal{P}'$.
To test whether a given $X$ is $s$-viable for $R$ with $s\geq 2$, we can just test whether $X'$ is $1$-viable for $R$ in $D'$.
Since $|X'| = d \cdot s$, this yields an algorithm in time $\Ocal(n^{d+2} \cdot k^{d\cdot s})$ for \tbdp.
\end{proof}

We now proceed to show that $(k,d,s)$-\tbdp is \textsf{FPT}, by providing a kernel with at most $d \cdot 2^{k-s} \cdot \binom{k}{s} + 2k$ vertices. We start with some definitions and technical lemmas.
Notice that any vertex in $D$ whose deletion disconnects more than $s$ pairs in the request set $R$ cannot be contained in any set $X$ that is $s$-viable for $R$.
Hence we make use of an operation to eliminate all such vertices from the input digraph while maintaining connectivity.
We remind the reader that, for a request set $R$, we denote by $S(R)$ the set of source vertices in $R$ and by $T(R)$ the set of target vertices in $R$ (cf. Definition~\ref{definition:requests}).
\begin{definition}[Non-terminal vertices]
Let $(D, R, k, d, s)$ be an instance of {\sc DEDP}. For a digraph $D'$ such that $V(D') \subseteq V(D)$, we define $V^*(D') = V(D') \setminus (S(R) \cup T(R))$.
\end{definition}
That is, $V^*(D)$ is the set of non-terminal (i.e., neither source nor target) vertices of $D$.

\begin{definition}[Congested vertex and \BcolName]\label{defintion:blocking_vertex}
Let $(D, R, k, d, s)$ be an instance of {\sc DEDP}.
For $X \subseteq V^*(D)$, we define $R_X$ as the subset of $R$ that is \emph{blocked} by $X$,
that is, there are no paths from $s$ to $t$ in $D - X$ for every $(s_i,t_i) \in R_X$.
We say that a vertex $v \in V^*(D)$ is an $(R,s)$-\emph{congested vertex} of $D$ if $|R_{\{v\}}| \geq s+1$.
The \emph{\BcolName of $R$} is the collection $\{B_1, \ldots, B_k\}$ where $B_i = \{v \in V^*(D) \mid (s_i, t_i) \in R_{\{v\}}\}$, for $i \in [k]$.
We say that $D$ is \emph{clean for $R$} and that $(D, R, k, d, s)$ is a \emph{clean instance} if there are no congested vertices in $V^*(D)$.
When $R$ and $s$ are clear from the context, we drop them from the notation.
\end{definition}
We use the following operation to eliminate congested vertices of $D$ while maintaining connectivity.
It is used, for instance, in~\cite{Chitnis.Hajiaghayi.Marx.11} (as the \emph{torso} operation) and in~\cite{Kratsch.Pilipczuk.Wahlstrom.15}.
\begin{definition}[Bypassing vertices and sets]
Let $D$ be a graph and $v \in V(D)$.
We refer to the following operation as \emph{bypassing $v$}: delete $v$ from $D$ and, for each $u \in N^-(v)$ add one edge from $u$ to each vertex $w \in N^+(v)$.
We denote by $D / v$ the graph generated by bypassing $v$ in $D$.
For a set of vertices $B \subseteq V(D)$, we denote by $D / B$ the graph generated by bypassing, in $D$, all vertices of $B$ in an arbitrary order.
\end{definition}
Figure~\ref{fig:bypassing} illustrates the bypass operation.
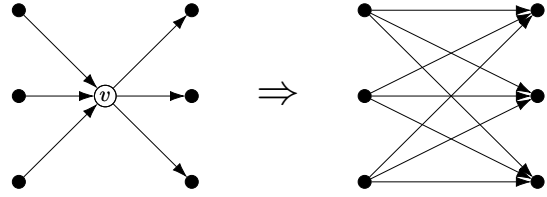
\begin{figure}[t]
\centering
\scalebox{.65}{
\begin{tikzpicture}[scale=1.75,-{Latex[length=2mm, width=2mm]},shorten >= .1cm]
\node[blackvertex, scale=.75] (P-a) at (-1,2) {};
\node[blackvertex, scale=.75] (P-b) at (-1,1) {};
\node[blackvertex, scale=.75] (P-c) at (-1,0) {};

\node[vertex, scale=1.35] (P-d) at (0,1) {$v$};

\node[blackvertex, scale=.75] (P-e) at (1,2) {};
\node[blackvertex, scale=.75] (P-f) at (1,1) {};
\node[blackvertex, scale=.75] (P-g) at (1,0) {};

\foreach \x/\y in {
a/d, b/d, c/d, d/e, d/f, d/g}
	\draw (P-\x) -- (P-\y);

\node (P-n) at (2,1) {\Huge$\Rightarrow$};

\node[blackvertex, scale=.75] (P-a) at (3,2) {};
\node[blackvertex, scale=.75] (P-b) at (3,1) {};
\node[blackvertex, scale=.75] (P-c) at (3,0) {};


\node[blackvertex, scale=.75] (P-e) at (5,2) {};
\node[blackvertex, scale=.75] (P-f) at (5,1) {};
\node[blackvertex, scale=.75] (P-g) at (5,0) {};

\foreach \x/\y in {
a/e, a/f, a/g, b/e, b/f, b/g, c/e, c/f, c/g}
	\draw [-{Latex[length=3mm, width=2mm]},] (P-\x) -- (P-\y);


\end{tikzpicture}%
}%
\caption{Bypassing a vertex $v$.}
\label{fig:bypassing}
\end{figure}

We restrict our attention to vertices in $V^*(D)$ in Definition~\ref{defintion:blocking_vertex} because we want to avoid bypassing source or target vertices, and work only with vertices inside $V^*(D)$.
Since $|S(R) \cup T(R)| \leq 2k$, we  show later that this incurs an additive term of $2k$ in the size of the constructed kernel.

In~\cite[Lemma 3.6]{Kratsch.Pilipczuk.Wahlstrom.15} the authors remark that the ending result of bypassing a set of vertices in a digraph does not depend on the order in which those vertices are bypassed.
Furthermore, bypassing a vertex of $D$ cannot generate a new congested vertex: if $u$ is a congested vertex of $D / v$, then $u$ is also a congested vertex of $D$, for any $v \in V(D) \setminus \{u\}$.
Thus any instance $(D, R, k, d, s)$ of \tbdp is equivalent to the instance $(D/v, R, k, d, s)$, if $v$ is a congested vertex of $D$, and arbitrarily bypassing a vertex of $D$ can only make the problem harder.
We formally state those observations below.
\begin{lemma}\label{lemma:bypass_equivalence}
Let $D$ be a digraph, $R$ be a request set with $R = \{(s_1, t_1), \ldots, (s_k, t_k)\}$, $s$ be an integer, $B$ be the set of $(R,s)$-congested vertices of $D$, and $D' = D / B$.
Then, with respect to $R$, $X$ is $s$-viable in $D$ if and only if $X$ is $s$-viable in $D'$.
\end{lemma}
\begin{proof}
Let $X$ be an $s$-viable set for $R$ in $D$. 
If $X \cap B \neq \emptyset$, then at least $s+1$ paths of any collection satisfying $R$ must intersect in $X$, contradicting our choice for $X$, and hence $X \subseteq V(D')$. 
Similarly, if $X \subseteq V(D')$ then $X \cap B = \emptyset$ and the sufficiency follows.
\end{proof}

Furthermore, from any solution for an instance resulting from bypassing a set of vertices in $V^*(D)$, we can construct a solution for the original instance in polynomial time by undoing the bypasses.
\begin{remark}\label{remark:bypass_harder}
Let $(D, R, k, d, s)$ be an instance of \tbdp and $Y \subseteq V^*(D)$.
If $\mathcal{P}$ is a solution for $(D/Y, R, k, d, s)$, then $(D, R, k, d, s)$ is positive and a solution can be constructed from $\mathcal{P}$ in polynomial time.
\end{remark}

The main ideas of the kernelization algorithm are the following.
Let $(D, R, k, d, s)$ be an instance of \tbdp and $\{B_1, \ldots, B_k\}$ be the \BcolName of $R$.
First, we show that, if $D$ is clean for $R$, there is an $i \in [k]$ such that $|V^*(D) - B_i| \geq |V^*(D)|(k-s)/k$ (Lemma~\ref{lemma:iteration_lemma}).
Then, we show that if $D$ is clean and sufficiently large, and $|R| = s+1$, then the instance is positive and a solution can be found in polynomial time (Lemma~\ref{lemma:algorithm_base_instance}).

Lemma~\ref{lemma:algorithm_base_instance} is used as the base case for our iterative algorithm.
We start with the first instance, say $(D, R, k, d, s)$, and proceed through $k-s+1$ iterations.
At each iteration, we will choose one path from some $s_i$ to its destination $t_i$ such that a large part of the graph remains unused by any of the pairs chosen so far (by Lemma~\ref{lemma:iteration_lemma}) and consider the request set containing only the remaining pairs for the next iteration.
We repeat this procedure until we arrive at an instance where the number of requests is exactly $s+1$, and show that if $n$ is large enough, then we can use Lemma~\ref{lemma:algorithm_base_instance} to output a solution for the last instance.
From this solution, we extract a solution for $(D, R, k, d, s)$ in polynomial time.

\begin{lemma}\label{lemma:iteration_lemma}
Let $(D, R, k, d, s)$ be an instance of {\sc DEDP}, $\{B_1, \ldots, B_k\}$ be the \BcolName of $R$, and $n^* = |V^*(D)|$.
If $D$ is clean, then there is an $i \in [k]$ such that $|V^*(D/B_i)| \geq n^*(k-s)/k$ and there is a path $P$ in $D / B_i$ from $s_i$ to $t_i$ such that $|V^*(P)| \leq |V^*(D/B_i)| / 2$.
\end{lemma}
\begin{proof}
First, notice that
$$\sum_{v \in V^*(D)}|R_v| = \sum_{v \in V^*(D)}|\{B_i \mid  v \in B_i\}| = \sum_{i \in [k]}|B_i|.$$
Now, if $|B_i| > n^* \cdot s/k$ for every $i \in [k]$, then there must be a vertex in $v$ such that $|R_v| \geq s+1$, as in this case $\sum_{i \in [k]}|B_i|
 > n^*\cdot s$, contradicting our assumption that $D$ is clean.
We conclude that there is an $i \in [k]$ such that $|B_i| \leq n^* \cdot s/k$ and thus $V^*(D/B_i) = n^* - |B_i| \geq n^*(k-s)/k$, as desired.

The result trivially follows if there is a path  $P$ from $s_i$ to $t_i$ in $D / B_i$ with $V^*(P) = \emptyset$.
Thus we can assume that every path from $s_i$ to $t_i$ in $D/B_i$ intersects $V^*(D/B_i)$ (see Figure~\ref{fig:iteration_lemma_2}).
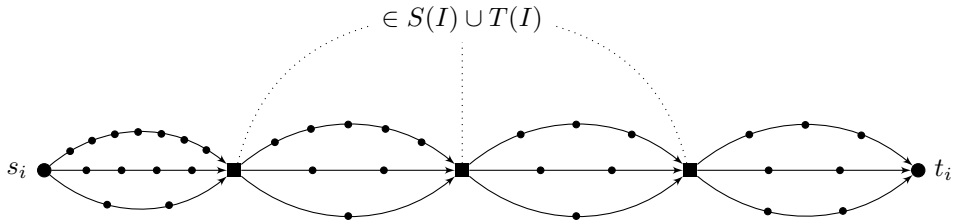
\begin{figure}[h!tb]
\centering
\begin{tikzpicture}[scale=1]

\foreach \i in {0,...,2} {
	\node [rectangle, fill=black,scale=.75] (P-v\i) at (2.5+3*\i,0) {};
}

\node (P-n) at ($(P-v1) + (0,1.5)$) {\textbf{$\in S(R) \cup T(R)$}};
\draw [dotted] (P-v0) to [out = 90, in = 180] (P-n);
\draw [dotted] (P-v1) to  (P-n);
\draw [dotted] (P-v2) to [out = 90, in = 0] (P-n);

\node[scale=.5,label=left:$s_i$, blackvertex] (P-s) at (0,0) {};
\node[scale=.5,label=right:$t_i$, blackvertex] (P-t) at (11.5,0) {};

\draw[edge] (P-s) to [bend left = 40]
node [scale=.25,blackvertex,pos=1/8] {}
node [scale=.25,blackvertex,pos=2/8] {}
node [scale=.25,blackvertex,pos=3/8] {}
node [scale=.25,blackvertex,pos=4/8] {}
node [scale=.25,blackvertex,pos=5/8] {}
node [scale=.25,blackvertex,pos=6/8] {}
node [scale=.25,blackvertex,pos=7/8] {}
(P-v0);

\draw[edge] (P-s) to
node [scale=.25,blackvertex,pos=1/5] {}
node [scale=.25,blackvertex,pos=2/5] {}
node [scale=.25,blackvertex,pos=3/5] {}
node [scale=.25,blackvertex,pos=4/5] {}
(P-v0);

\draw[edge] (P-s) to [bend right = 40]
node [scale=.25,blackvertex,pos=1/3] {}
node [scale=.25,blackvertex,pos=2/3] {}
(P-v0);

\draw[edge] (P-v0) to [bend left = 40]
node [scale=.25,blackvertex,pos=1/6] {}
node [scale=.25,blackvertex,pos=2/6] {}
node [scale=.25,blackvertex,pos=3/6] {}
node [scale=.25,blackvertex,pos=4/6] {}
node [scale=.25,blackvertex,pos=5/6] {}
(P-v1);

\draw[edge] (P-v0) to
node [scale=.25,blackvertex,pos=1/3] {}
node [scale=.25,blackvertex,pos=2/3] {}
(P-v1);

\draw[edge] (P-v0) to [bend right = 40]
node [scale=.25,blackvertex,pos=1/2] {}
(P-v1);

\draw[edge] (P-v1) to [bend left = 40]
node [scale=.25,blackvertex,pos=1/4] {}
node [scale=.25,blackvertex,pos=2/4] {}
node [scale=.25,blackvertex,pos=3/4] {}
(P-v2);

\draw[edge] (P-v1) to
node [scale=.25,blackvertex,pos=1/3] {}
node [scale=.25,blackvertex,pos=2/3] {}
(P-v2);

\draw[edge] (P-v1) to [bend right = 40]
node [scale=.25,blackvertex,pos=1/2] {}
(P-v2);

\draw[edge] (P-v2) to [bend left = 40]
node [scale=.25,blackvertex,pos=1/4] {}
node [scale=.25,blackvertex,pos=2/4] {}
node [scale=.25,blackvertex,pos=3/4] {}
(P-t);

\draw[edge] (P-v2) to
node [scale=.25,blackvertex,pos=1/3] {}
node [scale=.25,blackvertex,pos=2/3] {}
(P-t);

\draw[edge] (P-v2) to [bend right = 40]
node [scale=.25,blackvertex,pos=1/3] {}
node [scale=.25,blackvertex,pos=2/3] {}
(P-t);
\end{tikzpicture}
\caption{Three paths from $s_i$ to $t_i$ in $D/B_i$. Square vertices are used to identify vertices in $S(R) \cup T(R)$, which may not be bypassed.}
\label{fig:iteration_lemma_2}
\end{figure}
Let $X = (S(R) \cup T(R)) \setminus \{s_i, t_i\}$.
By Menger's Theorem and since no vertex in $V^*(D/B_i)$ intersects every path from $s_i$ to $t_i$, there are two internally disjoint paths $P_1$ and $P_2$ from $s_i$ to $t_i$ in $(D/B_i)/X$.
Without loss of generality, assume that $P_1$ is the shortest of those two paths, breaking ties arbitrarily.
Then $|V^*(P_1)| \leq |V^*(D/B_i)| /2$ since $P_1$ and $P_2$ are disjoint, and the result follows.
\end{proof}

Figure~\ref{fig:iteration_lemma_1} illustrates the procedure described in Lemma~\ref{lemma:iteration_lemma}.
We find a set $B_i$ containing at most $n^* \cdot s/k$ vertices, and bypass all of its vertices in any order.
Then we argue that a shortest path from $s_i$ to $t_i$ in $D/B_i$ avoids a large set of vertices in $D$.
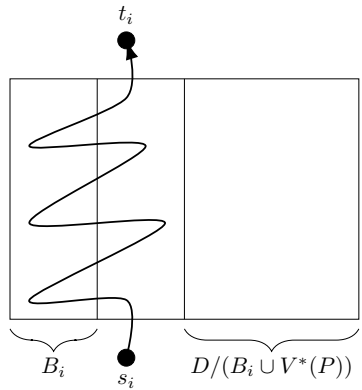
\begin{figure}[h!tb]
\centering
\scalebox{1}{
\begin{tikzpicture}[scale=1.1]
\node[blackvertex,label=below:{\Large $s_i$},scale=.4] (P-c) at (0.6,0.8) {};
\node[blackvertex, label={\Large $t_i$},scale=.4] (P-d) at (1,3.1) {};
\node (P-exit) at (0.6, 2.5) {};
\node[scale=.75] (P-b) at (1,1.4) {};

\draw[thick, -{Latex[length=2mm, width=1.5mm]}, shorten >= .15cm] plot[smooth] coordinates {(P-c)  (P-b) (-0.4,1.2) (0.8,2) (-0.4,2) (P-exit) (P-d)};

\draw (0.3, 2.7) -- (0.3,1);
\draw (1.2, 2.7) -- (1.2,1);
\draw [decorate,decoration={mirror,brace,amplitude=10pt}]
(1.2,0.9) -- (3,0.9) node [black,midway,yshift=-0.6cm] {$D/(B_i \cup V^*(P))$};
\draw (-0.6,1) rectangle (3,2.7);
\draw [decorate,decoration={mirror,brace,amplitude=7pt}]
(-0.6,0.9) -- (0.3,0.9) node [black,midway,yshift=-0.6cm] {$B_i$};
\end{tikzpicture} 
}%
\caption{A path $P$ from $s_i$ to $t_i$ avoiding a large part of $D$.}
\label{fig:iteration_lemma_1}
\end{figure}

\begin{lemma}\label{lemma:algorithm_base_instance}
Let $(D, R, k, d, s)$ be an instance of {\sc DEDP}, $m = |E(D)|$, and $n^* = V^*(D)$.
If $D$ is clean, $n^* \geq 2d(s+1)$, and $k = s+1$, then $(D, R, k, d, s)$ is positive and a solution can be found in time $\Ocal(k \cdot n(n+m))$.
\end{lemma}
\begin{proof}
Let $\{B_1, \ldots, B_k\}$ be the \BcolName of $R$ and $D'_i = D / B_i$, for $i \in [k]$.
By Lemma~\ref{lemma:iteration_lemma}, there is an $i \in [k]$ such that $|V^*(D'_i)| \geq n^*/(s+1)$ and a path $P$ from $s_i$ to $t_i$ such that $V^*(P) \leq |V^*(D'_i)|/2$.
Let $D_i = D'_i / V^*(P)$.
Now,
$$|V^*(D_i)| \ \geq\  \frac{|V^*(D'_i)|}{2}\ \geq \ \frac{n^*}{2(s+1)}$$
and since $|R \setminus \{(s_i, t_i)\}| = s$, we are free to choose arbitrarily any collection of paths satisfying $R \setminus \{(s_i, t_i)\}$ in $D_i$.
Reversing the bypasses done in $D$, this collection together with $P_i$ yields a collection of paths satisfying $R$ in $D$ such that all vertices in $V^*(D_i)$ are contained in at most $s$ of those paths. Since $n^* \geq 2d(s+1)$ by hypothesis, we have that $|V^*(D_i)| \geq d$ as required. We can generate the sets $B_i$ in time $\Ocal(n(n+m))$ by deleting a vertex of $D$ and testing for connectivity between $s_i$ and $t_i$.
Thus a solution can be found in time $\Ocal(k\cdot n (n+m))$, as desired.
\end{proof}
We are now ready to show the main ingredient of the algorithm: we provide a polynomial-time algorithm to solve large clean instances of the \tbd~problem.

\begin{theorem}\label{theorem:solving_large_instances_generalization}
Let $(D, R, k, d, s)$ be a clean instance of {\sc DEDP} with
$|V^*(D)| = n^* \geq d \cdot 2^{k-s} \cdot \binom{k}{s}$.
Then $(D, R, k, d, s)$ is positive and a solution can be found in time $\Ocal(k\cdot n^2(n+m))$.
\end{theorem}
\begin{proof}
Let $\mathcal{B}_0 = \{B_1, \ldots, B_k\}$ be the \BcolName of $R$.
We consider $\mathcal{B}_0$ to be sorted in non-decreasing order by the size of its elements and, by rearranging $R$ if needed, we assume that this order agrees with $R$.
For $i \in [k-(s+1)]$, we construct a sequence of sets $\{D_i, \mathcal{B}_i, \mathcal{P}_i\}$ where $n^*_i = |V^*(D_i)|$ and

\begin{enumerate}
	\item[(i)] $\mathcal{B}_{i} = \{B_{i+1}, \ldots, B_{k}\}$;
	\item[(ii)] $\mathcal{P}_{i}$ is a collection of paths $\{P_1, P_2, \ldots, P_i\}$ such that $P_j$ is a path from $s_j$ to $t_j$ in $D_j$, for $j \in [i]$; and
	\item[(iii)] $n^*_{i-1}$ is large enough to guarantee that we can find a path from $s_{i}$ to $t_{i}$ avoiding a large part of $D_{i-1}$. Formally, we want that
$$n^*_{i}\ \geq\ n^*_0  \cdot \frac{(k-s)(k-s-1)\cdots(k-s-i+1)}{2^i \cdot k (k-1)\cdots(k-i+1)}.$$
\end{enumerate}


We begin with $D_0 = D$, $n^*_0 = n^*$, and $\mathcal{P}_0 = \emptyset$.
Let $D'_1 = D_0/B_1$. By applying Lemma~\ref{lemma:iteration_lemma} with input $(D_0, R, k, d, s)$, we conclude that $|V^*(D'_1)| \geq n^*(k-s)/k$ and there is a path $P_1$ from $s_1$ to $t_1$ in $D'_1$ with $|V^*(P_1)| \leq |V^*(D'_1)|/2$.
Let $D_1 = D'_1 / V^*(P_1)$ and $\mathcal{P}_1 = \{P_1\}$.
Now,
$$n^*_1\ \geq\ \frac{|V^*(D'_1)|}{2} \ \geq  \ \frac{n^*_0(k-s)}{2k}$$
and conditions \lip{(i)}, \lip{(ii)}, and \lip{(iii)} above hold for $(D_1, \mathcal{B}_1, \mathcal{P}_1)$. Assume that $i-1$ triples have been chosen in this way.



As before, we assume that $\mathcal{B}_{i-1}$ is sorted in non-increasing order by the size of its elements, and that this order agrees with $R \setminus R_{i-1}$.
Furthermore, as $D_0$ is clean, so is $D_{i-1}$.

Let $D'_i = D_{i-1} / B_i$. Applying Lemma~\ref{lemma:iteration_lemma} with input $(D_{i-1}, R \setminus R_{i-1}, k - i + 1, d, s)$, we conclude that $|V^*(D'_i)| \geq n^*_i(k  - i + 1 - s)/(k-i+1)$ and there is a path $P_i$ from $s_i$ to $t_i$ in $D'_i$ with $|V^*(P_i)| \leq |V^*(D'_i)|/2$.
Let $\mathcal{P}_i = \mathcal{P}_{i-1} \cup \{P_i\}$ and $D_i = D'_i / B_i$.
Then
$$
n^*_i \ \geq\ \: n^*_{i-1} \cdot \frac{k - i + 1 - s}{2 (k-i+1)}
$$
and by our assumption that \lip{(iii)} holds for $n_{i-1}$ it follows that
\begin{align*}
\hspace{.8cm}n^*_i\ \geq\ &\: n^*_0 \cdot \frac{(k-s)(k-s-1)\cdots (k-s-i+2)}{2^{i-1}k(k-1)\cdots (k-i+2)}\cdot\left(\frac{k-s-i+1}{2(k-i+1)}\right)\\
\ =\ &\: n^*_0 \cdot\frac{(k-s)(k-s-1)\cdots(k-s-i+1)}{2^i \cdot k (k-1)\cdots(k-i+1)},
\end{align*}

as desired and thus \lip{(i)}, \lip{(ii)}, and \lip{(iii)} hold for $(D_i, \mathcal{B}_i, \mathcal{P}_i)$.
The algorithm ends after iteration $k-(s+1)$.
Following this procedure, we construct the collection $\mathcal{P}_{k-(s+1)} = \{P_1, P_2, \ldots, P_{k-(s+1)}\}$ satisfying \lip{(ii)} and the graph $D_{k-(s+1)}$ with $n_{k-(s+i)}$ satisfying \lip{(iii)}.
Noticing that $|R - R_{k-(s+1)}| = s+1$ (that is, only $s+1$ pairs in $R$ are not accounted for in $\mathcal{P}_{k-(s+1)}$), it remains to show that our choice for $n^*$ is large enough so that we are able to apply Lemma~\ref{lemma:algorithm_base_instance} on the instance $(D_{k-(s+1)}, R - R_{k-(s+1)}, s+1, d, s)$ of \tbdp.
That is, we want that $n^*_{k-(s+1)} \geq 2d(s+1)$. By \lip{(iii)} it is enough to show that
$$n^*_{k-(s+1)}\ \geq\ n^*_0 \cdot \frac{(k-s)(k-s-1) \cdots 3 \cdot 2}{2^{k-(s+1)} \cdot k(k-1) \cdots (s+3)(s+2)}\ \geq\ 2d \cdot (s+1),$$
and rewriting both sides of the fraction as $k!$ and $k!/(s+1)!$, respectively, we get
$$
n_0 \cdot \frac{(k-s)!}{2^{k-(s+1)}}\ \geq \ 2d \cdot (s+1) \cdot \frac{k!}{(s+1)!}\ =\ \frac{2d\cdot k!}{s!},
$$
which holds for
$$ n_0 \ \geq \ \left(\frac{2^{k-(s+1)}\cdot 2d \cdot (s+1)}{(s+1)!}\right)\cdot\left(\frac{k!}{(k-s)!}\right)
\ = \ d \cdot 2^{k-s}\cdot \binom{k}{s},
$$
as desired.

Applying Lemma~\ref{lemma:algorithm_base_instance} with input $(D_{k-(s+1)},R \setminus R_{k-(s+1)}, s+1, d,s)$ yields a collection $\hat{\mathcal{P}}$ satisfying $R - R_{k-(s+1)}$ and a set $X \subseteq V(D)$ of size $d$ (since $V(D_{k-(s+1)}) \subseteq V(D)$) such that $X$ is disjoint from all paths in $\mathcal{P}_{k-(s+1)}$, since all vertices in $V^*(P)$ were bypassed in $D_{k-(s+1)}$ for every $P \in \mathcal{P}_{k-(s+1)}$, and all vertices in $X$ occur in at most $s$ paths of $\hat{\mathcal{P}}$.
We can construct a collection of paths satisfying $R$ from $\hat{\mathcal{P}} \cup \mathcal{P}_{k-(s+1)}$ by reversing all the bypasses done in $D$ and connecting appropriately the paths in the collections (see Remark~\ref{remark:bypass_harder}).
We output this newly generated collection as a solution for $(D, R, k, d, s)$.

For the running time, let $m = |(E(D)|$. We need time $\Ocal(k \log k)$ to order the elements of $\mathcal{B}_0$, time $\Ocal(k \cdot n(n+m))$ to find the sets $B_i$, for $i \in [k]$, and  $\Ocal(n+m)$ to find each of the paths $\{P_1, \ldots, P_k\}$. Hence the algorithm runs in time $\Ocal(k \cdot n^2(n+m))$.
\end{proof}

We acknowledge that it is possible to prove Theorem~\ref{theorem:solving_large_instances_generalization} without using Lemma~\ref{lemma:algorithm_base_instance} by stopping the iteration at the digraph $D_{k-s}$ instead of $D_{k-s-1}$.
However we believe it is easier to present the proof of Theorem~\ref{theorem:solving_large_instances_generalization} by having separate proofs for the iteration procedure (Lemma~\ref{lemma:iteration_lemma}) that aims to generate an instance of DEDP for which we can apply our base case (Lemma~\ref{lemma:algorithm_base_instance}).

Since any instance can be made clean in polynomial time, the kernelization algorithm for $(k,d,s)$-\tbdp follows easily. Given an instance $(D, R, k, d, s)$, we bypass all congested vertices of $D$ to generate $D'$. If $|V^*(D')|$ is large enough to apply Theorem~\ref{theorem:solving_large_instances_generalization}, the instance is positive and we can find a solution in polynomial time.
Otherwise, we generated an equivalent instance $(D', R, k, d, s)$ with $|V(D')|$ bounded from above by a function depending on $k, d$, and $s$ only.
As we restrict $|S(R) \cup T(R)| \leq 2k$, if $D$ is clean and $V(D) \geq d \cdot 2^{k-s} \cdot \binom{k}{s} + 2k$ we get the desired bound for $|V^*(D)|$. Thus, the following  is  a direct corollary of Theorem~\ref{theorem:solving_large_instances_generalization}.

\begin{theorem}\label{corollary:kernel_for_k_d_s_problem}
There is a kernelization algorithm running in time $\Ocal(k \cdot n^2(n+m))$ that, given an instance $(D, R, k, d, s)$ of {\sc DEDP}, outputs either a solution for the instance or an equivalent instance $(D', R, k, d, s)$ with
$|V(D')|\ \leq\ d \cdot 2^{k-s} \cdot \binom{k}{s} + 2k.$
\end{theorem}

\section{Concluding remarks}\label{subsection:poly_kernel_open}

We introduced the \tbd problem and provided a number of hardness and algorithmic results, summarized in Table~\ref{table:summary_of_results}. Several questions remain open.

We showed that \tbdp is \textsf{NP}-complete for every fixed $k\geq 3$ and $s = 1$ and we leave open the question of whether \textsc{DEDP} is \textsf{NP}-complete for every fixed $k \geq 3$ and fixed $s \geq 2$ such that $s < k$.
The reduction used in the proof of Theorem~\ref{theorem:npc_for_k} also shows that \textsc{DEDP} is as hard as \textsc{Directed Disjoint Paths with Congestion}. 
Although it is not known whether the latter is \textsf{NP}-complete for some fixed value of $k$ when the congestion is at least two, a positive answer to this question would also imply the \textsf{NP}-completeness of \textsc{DEDP} for some fixed $k$ and $s \geq 1$ with $s < k$.
Giannopoulou et al.~\cite[Conjecture 1.3]{giannopoulou2020canonical} conjectured that \textsc{DDPC} is solvable in polynomial time for every fixed $k \geq 2$ and fixed congestion $s = 2$.

We provided an algorithm running in time $\Ocal(n^{d+2} \cdot k^{d\cdot s})$ to solve the problem. This algorithm tests all partitions of a given $X \subseteq V(D)$ in search for one that respects some properties.
Since there are at most $\binom{n}{d}$ subsets of $V(D)$ of size $d$, this yields an \textsf{XP} algorithm.
The second term on the time complexity comes from the number of partitions of $X$ we need to test. The problem may become easier if $X$ is already given or, similarly,  if $d$ is a constant.
In other words, is the $(s)$-\tbdp problem \textsf{FPT} for fixed $d$?

Our main result is a kernel with at most $d \cdot 2^{k-s} \cdot \binom{k}{s} + 2k$ vertices. The natural question is whether the problem admits a polynomial kernel with parameters $k$, $d$, and $s$, or even for fixed $s$.
Notice that if there is a constant $\ell$ such that $k-s = \ell$, then the size of the kernel is $d \cdot 2^{\ell} \cdot k^{\ell}$, which is polynomial on $d$ and $k$.
The case $s=0$ is also particularly interesting, as \tbdp with $s=0$ is equivalent to the \textsc{Steiner Network} problem. In this case, we get a kernel of size at most $d \cdot 2^{k} + 2k$.

While we do not know whether $(k,d,s)$-\tbdp admits a polynomial kernel, at least we are able to prove that a negative answer for $s=0$ is enough to show that $(k,d,s)$-\tbdp~is unlikely to admit a polynomial kernel for any value of $s \geq 1$ when $k$ is ``far'' from $s$, via the following polynomial time and parameter reduction.

\begin{remark}\label{remark-kernel}
For any instance $(D, R, k, d, 0)$ of {\sc DEDP} and integer $s > 0$, one can construct in polynomial time an equivalent instance $(D, R', k', d, s)$ of {\sc DEDP} with $k' = k \cdot (d \cdot s + 1)$.
\end{remark}
\begin{proof}
For a request set $R$ in $D$, let $R'$ be the request set in $D$ formed by $d \cdot s +1$ copies of each pair in $R$ and let $k' = k \cdot (d \cdot s +1)$.
We claim that an instance $(D, R, k, d, 0)$ of \tbdp~is positive if and only if the associated instance $(D, R', k', d, s)$, also of \tbdp, is positive.

From any solution $\mathcal{P}$ for the first instance, we can construct a solution for the second by taking $d\cdot s + 1$ copies of each path in $\mathcal{P}$ and thus the necessity holds.
For the sufficiency, let $X$ be an $s$-viable set for $(D, R', k', d, s)$ with certifying collection $\mathcal{P'}$.
By the construction of $R'$ and since at most $d \cdot s$ paths in $\mathcal{P'}$ can intersect $X$, we conclude that there is path $P \in \mathcal{P'}$ from $s$ to $t$ in $D - X$ for each pair $(s, t) \in R$.
Choosing all such paths we construct a collection $\mathcal{P}$ satisfying $R$ in $D-X$, and the remark follows.
\end{proof}

In the undirected case, the \textsc{Steiner Tree} problem
is unlikely to admit a polynomial kernel  parameterized by $k$ and $c$, with $c = n -d$ (in other words, the size of the solution); a simple proof for this result can be found in~\cite[Chapter 15]{CyganFKLMPPS15}.  Even if we consider a stronger parameter (that is, $d$ instead of $c$), dealing with directed graphs may turn the problem much harder. We also remark that the problem admits a polynomial kernel in the undirected case if the input graph is planar~\cite{Pilipczuk:2018:NSS:3266298.3239560}. It may also be the case for directed graphs.

\medskip
\noindent \textbf{Acknowledgement}. We would like to thank the anonymous reviewers for helpful and thorough comments that improved the presentation of the manuscript, in particular for pointing out a wrong argument in the proof of Theorem~\ref{theorem:npc_for_k} that made us claim, in the conference version of this article, that \textsc{DEDP} is \textsf{NP}-complete for every fixed $k \geq 3$ and $s \geq 1$ such that $s < k$. As mentioned above, this is still open.




\bibliography{references}

\begin{thebibliography}{10}

\bibitem{ADLER2007718}
Isolde Adler.
\newblock Directed tree-width examples.
\newblock {\em Journal of Combinatorial Theory, Series B}, 97(5):718 -- 725,
  2007.
\newblock \href {https://doi.org/https://doi.org/10.1016/j.jctb.2006.12.006}
  {\path{doi:https://doi.org/10.1016/j.jctb.2006.12.006}}.

\bibitem{Amiri2016RoutingWC}
Saeed~Akhoondian Amiri, Stephan Kreutzer, D{\'{a}}niel Marx, and Roman
  Rabinovich.
\newblock Routing with congestion in acyclic digraphs.
\newblock {\em Information Processing Letters}, 151, 2019.
\newblock \href {https://doi.org/10.1016/j.ipl.2019.105836}
  {\path{doi:10.1016/j.ipl.2019.105836}}.

\bibitem{AraujoCLSSS18}
J{\'{u}}lio Ara{\'{u}}jo, Victor~A. Campos, Carlos Vin{\'{\i}}cius G.~C. Lima,
  Vin{\'{\i}}cius~Fernandes dos Santos, Ignasi Sau, and Ana Silva.
\newblock {Dual Parameterization of Weighted Coloring}.
\newblock In {\em Proc. of the 13th International Symposium on Parameterized
  and Exact Computation, (IPEC)}, volume 115 of {\em LIPIcs}, pages
  12:1--12:14, 2018.
\newblock \href {https://doi.org/10.4230/LIPIcs.IPEC.2018.12}
  {\path{doi:10.4230/LIPIcs.IPEC.2018.12}}.

\bibitem{Classes.Directed.Graphs}
Jørgen Bang-Jensen and Gregory Gutin.
\newblock {\em Classes of Directed Graphs}.
\newblock Springer Monographs in Mathematics, 2018.

\bibitem{BasavarajuFRS16}
Manu Basavaraju, Mathew~C. Francis, M.~S. Ramanujan, and Saket Saurabh.
\newblock Partially polynomial kernels for set cover and test cover.
\newblock {\em {SIAM} Journal on Discrete Mathematics}, 30(3):1401--1423, 2016.
\newblock \href {https://doi.org/10.1137/15M1039584}
  {\path{doi:10.1137/15M1039584}}.

\bibitem{math.mag.85.4.252}
Khristo~N. Boyadzhiev.
\newblock {Close Encounters with the Stirling Numbers of the Second Kind}.
\newblock {\em Mathematics Magazine}, 85(4):252--266, 2012.
\newblock \href {https://doi.org/10.4169/math.mag.85.4.252}
  {\path{doi:10.4169/math.mag.85.4.252}}.

\bibitem{LAGOS19}
Victor Campos, Raul Lopes, Ana~Karolinna Maia, and Ignasi Sau.
\newblock {Adapting The Directed Grid Theorem into an {FPT} Algorithm}.
\newblock In {\em Proc. of the X Latin and American Algorithms, Graphs and
  Optimization Symposium (LAGOS)}, volume 346 of {\em ENTCS}, pages 229--240,
  2019.
\newblock \href {https://doi.org/10.1016/j.entcs.2019.08.021}
  {\path{doi:10.1016/j.entcs.2019.08.021}}.

\bibitem{Chitnis.Hajiaghayi.Marx.11}
Rajesh~Hemant Chitnis, MohammadTaghi Hajiaghayi, and D{\'{a}}niel Marx.
\newblock Fixed-parameter tractability of directed multiway cut parameterized
  by the size of the cutset.
\newblock {\em {SIAM} Journal on Computing}, 42(4):1674--1696, 2013.
\newblock \href {https://doi.org/10.1137/12086217X}
  {\path{doi:10.1137/12086217X}}.

\bibitem{ChorFJ04}
Benny Chor, Mike Fellows, and David~W. Juedes.
\newblock {Linear Kernels in Linear Time, or How to Save $k$ Colors in $O(n^2)$
  Steps}.
\newblock In {\em Proc. of the 30th International Workshop on Graph-Theoretic
  Concepts in Computer Science (WG)}, volume 3353 of {\em LNCS}, pages
  257--269, 2004.
\newblock \href {https://doi.org/10.1007/978-3-540-30559-0_22}
  {\path{doi:10.1007/978-3-540-30559-0_22}}.

\bibitem{Cook.Seymour.2003}
William Cook and Paul Seymour.
\newblock Tour merging via branch-decompositions.
\newblock {\em INFORMS Journal on Computing}, 15:233--248, 2003.
\newblock \href {https://doi.org/10.1287/ijoc.15.3.233.16078}
  {\path{doi:10.1287/ijoc.15.3.233.16078}}.

\bibitem{COURCELLE199012}
Bruno Courcelle.
\newblock {The monadic second-order logic of graphs. I. Recognizable sets of
  finite graphs}.
\newblock {\em Information and Computation}, 85(1):12--75, 1990.
\newblock \href {https://doi.org/10.1016/0890-5401(90)90043-H}
  {\path{doi:10.1016/0890-5401(90)90043-H}}.

\bibitem{CyganFKLMPPS15}
Marek Cygan, Fedor~V. Fomin, Lukasz Kowalik, Daniel Lokshtanov, D{\'{a}}niel
  Marx, Marcin Pilipczuk, Michał Pilipczuk, and Saket Saurabh.
\newblock {\em Parameterized Algorithms}.
\newblock Springer, 2015.
\newblock \href {https://doi.org/10.1007/978-3-319-21275-3}
  {\path{doi:10.1007/978-3-319-21275-3}}.

\bibitem{6686155}
Marek Cygan, Daniel Marx, Marcin Pilipczuk, and Michał Pilipczuk.
\newblock {The Planar Directed $k$-Vertex-Disjoint Paths Problem Is
  Fixed-Parameter Tractable}.
\newblock In {\em Proc. of the IEEE 54th Annual Symposium on Foundations of
  Computer Science (FOCS)}, volume~1, pages 197--206, 2013.
\newblock \href {https://doi.org/10.1109/FOCS.2013.29}
  {\path{doi:10.1109/FOCS.2013.29}}.

\bibitem{Demaine:2005:SPA:1101821.1101823}
Erik Demaine, Fedor~V. Fomin, Mohammadtaghi Hajiaghayi, and Dimitrios~M.
  Thilikos.
\newblock {Subexponential parameterized algorithms on bounded-genus graphs and
  $H$-minor-free graphs}.
\newblock {\em Journal of the ACM}, 52(6):866--893, 2005.
\newblock \href {https://doi.org/10.1145/1101821.1101823}
  {\path{doi:10.1145/1101821.1101823}}.

\bibitem{DF13}
Rod Downey and Michael.~R. Fellows.
\newblock {\em Fundamentals of Parameterized Complexity}.
\newblock Texts in Computer Science. Springer, 2013.
\newblock \href {https://doi.org/10.1007/978-1-4471-5559-1}
  {\path{doi:10.1007/978-1-4471-5559-1}}.

\bibitem{DuhF97}
Rong{-}chii Duh and Martin F{\"{u}}rer.
\newblock {Approximation of $k$-Set Cover by Semi-Local Optimization}.
\newblock In {\em Proc. of the 29th Annual {ACM} Symposium on the Theory of
  Computing (STOC)}, pages 256--264, 1997.
\newblock \href {https://doi.org/10.1145/258533.258599}
  {\path{doi:10.1145/258533.258599}}.

\bibitem{DBLP:conf/esa/EdwardsMW17}
Katherine Edwards, Irene Muzi, and Paul Wollan.
\newblock Half-integral linkages in highly connected directed graphs.
\newblock In {\em Proc. of the 25th Annual European Symposium on Algorithms
  (ESA), 2017}, pages 36:1--36:12, 2017.
\newblock \href {https://doi.org/10.4230/LIPIcs.ESA.2017.36}
  {\path{doi:10.4230/LIPIcs.ESA.2017.36}}.

\bibitem{feldmann_et_al:LIPIcs:2016:6306}
Andreas~Emil Feldmann and D{\'a}niel Marx.
\newblock {The complexity landscape of fixed-parameter directed Steiner network
  problems}.
\newblock In {\em Proc. of the 43rd International Colloquium on Automata,
  Languages, and Programming (ICALP)}, volume~55, pages 27:1--27:14, 2016.
\newblock \href {https://doi.org/10.4230/LIPIcs.ICALP.2016.27}
  {\path{doi:10.4230/LIPIcs.ICALP.2016.27}}.

\bibitem{FORTUNE1980111}
Steven Fortune, John Hopcroft, and James Wyllie.
\newblock The directed subgraph homeomorphism problem.
\newblock {\em Theoretical Computer Science}, 10(2):111--121, 1980.
\newblock \href {https://doi.org/https://doi.org/10.1016/0304-3975(80)90009-2}
  {\path{doi:https://doi.org/10.1016/0304-3975(80)90009-2}}.

\bibitem{giannopoulou2020canonical}
Archontia~C. Giannopoulou, Ken{-}ichi Kawarabayashi, Stephan Kreutzer, and
  O{-}joung Kwon.
\newblock The canonical directed tree decomposition and its applications to the
  directed disjoint paths problem, 2020.
\newblock \href {http://arxiv.org/abs/2009.13184v1}
  {\path{arXiv:2009.13184v1}}.

\bibitem{Johnson.Robertson.Seymour.Thomas.01}
Thor Johnson, Neil Robertson, Paul Seymour, and Robin Thomas.
\newblock Directed tree-width.
\newblock {\em Journal of Combinatorial Theory, Series B}, 82(01):138--154,
  2001.
\newblock \href {https://doi.org/10.1006/jctb.2000.2031}
  {\path{doi:10.1006/jctb.2000.2031}}.

\bibitem{10.1007/978-3-642-40450-4_57}
Mark Jones, Daniel Lokshtanov, M.~S. Ramanujan, Saket Saurabh, and Ond{\v{r}}ej
  Such{\'y}.
\newblock {Parameterized Complexity of Directed Steiner Tree on Sparse Graphs}.
\newblock {\em SIAM Journal on Discrete Mathematics}, 31(2):1294--1327, 2017.
\newblock \href {https://doi.org/10.1137/15M103618X}
  {\path{doi:10.1137/15M103618X}}.

\bibitem{KawarabayashiKK14}
Ken{-}ichi Kawarabayashi, Yusuke Kobayashi, and Stephan Kreutzer.
\newblock An excluded half-integral grid theorem for digraphs and the directed
  disjoint paths problem.
\newblock In {\em Proc. of the 46th ACM Symposium on Theory of Computing
  (STOC)}, pages 70--78, 2014.
\newblock \href {https://doi.org/10.1145/2591796.2591876}
  {\path{doi:10.1145/2591796.2591876}}.

\bibitem{Kawarabayashi:2015:DGT:2746539.2746586}
Ken-ichi Kawarabayashi and Stephan Kreutzer.
\newblock {The Directed Grid Theorem}.
\newblock In {\em Proc. of the 47th Annual ACM Symposium on Theory of Computing
  (STOC)}, pages 655--664, 2015.
\newblock \href {https://doi.org/10.1145/2746539.2746586}
  {\path{doi:10.1145/2746539.2746586}}.

\bibitem{Kratsch.Pilipczuk.Wahlstrom.15}
Stefan Kratsch, Marcin Pilipczuk, Michał Pilipczuk, and Magnus Wahlström.
\newblock Fixed-parameter tractability of multicut in directed acyclic graphs.
\newblock {\em SIAM Journal on Discrete Mathematics}, 29(1):122--144, 2015.
\newblock \href {https://doi.org/10.1137/120904202}
  {\path{doi:10.1137/120904202}}.

\bibitem{Lynch:1975:ETP:1061425.1061430}
James~F. Lynch.
\newblock The equivalence of theorem proving and the interconnection problem.
\newblock {\em ACM SIGDA Newsletter}, 5(3):31--36, 1975.
\newblock \href {https://doi.org/10.1145/1061425.1061430}
  {\path{doi:10.1145/1061425.1061430}}.

\bibitem{Menger1927}
Karl Menger.
\newblock Zur allgemeinen kurventheorie.
\newblock {\em Fundamenta Mathematicae}, 10(1):96--115, 1927.
\newblock URL: \url{http://eudml.org/doc/211191}.

\bibitem{Molle2008}
Daniel M{\"o}lle, Stefan Richter, and Peter Rossmanith.
\newblock Enumerate and expand: improved algorithms for connected vertex cover
  and tree cover.
\newblock {\em Theory of Computing Systems}, 43(2):234--253, 2008.
\newblock \href {https://doi.org/10.1007/s00224-007-9089-3}
  {\path{doi:10.1007/s00224-007-9089-3}}.

\bibitem{Pilipczuk:2018:NSS:3266298.3239560}
Marcin Pilipczuk, Micha\l\ Pilipczuk, Piotr Sankowski, and Erik Jan~Van
  Leeuwen.
\newblock {Network Sparsification for Steiner Problems on Planar and
  Bounded-Genus Graphs}.
\newblock {\em ACM Transactions on Algorithms}, 14(4):53:1--53:73, 2018.
\newblock \href {https://doi.org/10.1145/3239560} {\path{doi:10.1145/3239560}}.

\bibitem{Reed.99}
Bruce Reed.
\newblock Introducing directed tree-width.
\newblock {\em Electronic Notes in Discrete Mathematics}, 3:222--229, 1999.
\newblock \href {https://doi.org/10.1016/S1571-0653(05)80061-7}
  {\path{doi:10.1016/S1571-0653(05)80061-7}}.

\bibitem{Robertson.Seymour.86}
Neil Robertson and Paul Seymour.
\newblock {Graph minors. V. Excluding a planar graph}.
\newblock {\em Journal of Combinatorial Theory, Series B}, 41(01):92--114,
  1986.
\newblock \href {https://doi.org/10.1016/0095-8956(86)90030-4}
  {\path{doi:10.1016/0095-8956(86)90030-4}}.

\bibitem{ROBERTSON199565}
Neil Robertson and Paul Seymour.
\newblock {Graph minors. XIII. The disjoint paths problem}.
\newblock {\em Journal of Combinatorial Theory, Series B}, 63(1):65--110, 1995.
\newblock \href {https://doi.org/doi.org/10.1006/jctb.1995.1006}
  {\path{doi:doi.org/10.1006/jctb.1995.1006}}.

\bibitem{Sch94}
Alexander Schrijver.
\newblock Finding $k$ disjoint paths in a directed planar graph.
\newblock {\em SIAM Journal on Computing}, 23(4):780--788, 1994.
\newblock \href {https://doi.org/10.1137/S0097539792224061}
  {\path{doi:10.1137/S0097539792224061}}.

\bibitem{Slivkins.03}
Aleksandrs Slivkins.
\newblock Parameterized tractability of edge-disjoint paths on directed acyclic
  graphs.
\newblock {\em SIAM Journal on Discrete Mathematics}, 24(1):146--157, 2010.
\newblock \href {https://doi.org/10.1137/070697781}
  {\path{doi:10.1137/070697781}}.

\bibitem{Thomassen91}
Carsten Thomassen.
\newblock Highly connected non-2-linked digraphs.
\newblock {\em Combinatorica}, 11(4):393--395, 1991.
\newblock \href {https://doi.org/10.1007/BF01275674}
  {\path{doi:10.1007/BF01275674}}.

\end{thebibliography}

\end{document}